\documentclass[10pt,conference,romanappendices]{IEEEtran}
\usepackage[cmex10]{amsmath}
\usepackage{amsfonts}
\usepackage{amssymb}
\usepackage{wrapfig}
\usepackage{graphicx}
\usepackage{enumerate}
\usepackage{multirow}
\usepackage{xspace}
\usepackage{cancel}
\usepackage{xcolor,colortbl}

\newtheorem{example}{Example}
\newtheorem{definition}{Definition}
\newtheorem{proposition}{Proposition}
\newenvironment{proof}{\hspace{8pt}\ti{Proof:}}{$\blacksquare$}

\newtheorem{remark}{Remark}

\newcommand{\prd}{\ti{PrvRed}\xspace}
\newcommand{\lrn}{\ti{Lrn}\xspace}
\newcommand{\del}{\,\|\,}

\newcommand{\apqe}{\mbox{$\mi{DS}$-$\mi{PQE}$}\xspace}

\newcommand{\Apqe}{\mbox{$\mi{DS}$-$\mi{PQE}^+$}\xspace}

\newcommand{\olds}[4]{\mbox{$(\prob{#1}{#2},\pnt{#3})~\rightarrow #4$}}

\newcommand{\Ods}[5]{\mbox{$(\prob{#1}{#2},\pnt{#3},#4)~\rightarrow #5$}}

\newcommand{\ods}[3]{\mbox{$(\pnt{#1},#2)~\rightarrow #3$}}

\newcommand{\oDs}[3]{\mbox{$(#1,#2)~\rightarrow #3$}}
\newcommand{\noDs}[3]{$(#1,#2)~\rightarrow #3$}

\newcommand{\bm}[1]{{\mbox{\boldmath $#1$}}}

\newcommand{\ks}{\mbox{$\xi$}\xspace}
\newcommand{\di}[1]{\mbox{$\mi{Diam}(#1)$}\xspace}

\newcommand{\pnt}[1]{{\mbox{$\vec{#1}$}}}
\newcommand{\ppnt}[2]{{\mbox{$\vec{#1}_{#2}$}}}
\newcommand{\cof}[2]{\mbox{$#1_{\vec{#2}}$}}

\newcommand{\V}[1]{\mbox{$\mathit{Vars}(#1)$}}
\newcommand{\Va}[1]{\mbox{$\mi{Vars}(\vec{#1})$}}

\newcommand{\s}[1]{\mbox{$\{#1\}$}}

\newcommand{\nGz}[2]{$G_{non-\{z\}}$}

\newcommand{\prr}[1]{\mi{Prev}(\boldsymbol{q})}

\newcommand{\mi}[1]{\mathit{#1}}
\newcommand{\ti}[1]{\textit{#1}}
\newcommand{\tb}[1]{\textbf{#1}}

\newcommand{\ttt}{\>\>\>}
\newcommand{\tttt}{\>\>\>\>}

\newcommand{\Tt}{\>\>}
\newcommand{\Sub}[2]{\mbox{$\mi{#1}_\mi{#2}$}}

\newcommand{\Sup}[2]{\mbox{$#1^\mi{#2}$}}

\newcommand{\prob}[2]{\mbox{$\exists{#1} [#2]$}}

\newcommand{\ecnf}{\ensuremath{\exists\mathrm{CNF}}\xspace}
\newcommand{\Comment}[1]{}

\newcommand{\pqe}[4]{$#2 \wedge \prob{#1}{#4} \equiv \prob{#1}{#3 \wedge #4}$}

\begin{document}

\title{Partial Quantifier Elimination With Learning}


\author{\IEEEauthorblockN{Eugene Goldberg} 
\IEEEauthorblockA{
email:
eu.goldberg@gmail.com}}


\maketitle

\begin{abstract}
  We consider a modification of the Quantifier Elimination (QE)
  problem called Partial QE (PQE).  In PQE, only a small part of the
  formula is taken out of the scope of quantifiers. The appeal of PQE
  is that many verification problems, e.g. equivalence checking and
  model checking, reduce to PQE and the latter is much easier than
  complete QE. Earlier, we introduced a PQE algorithm based on the
  machinery of D-sequents. A D-sequent is a record stating that a
  clause is \textit{redundant} in a quantified CNF formula in a
  specified subspace. To make this algorithm efficient, it is
  important to reuse learned D-sequents. However, reusing D-sequents
  is not as easy as conflict clauses in SAT-solvers because redundancy
  is a \textit{structural} rather than a semantic property.
  In~\cite{qe_learn}, we modified the definition of D-sequents to
  enable their \textit{safe} reusing. In this paper, we present a PQE
  algorithm based on new D-sequents. It is different from its
  predecessor in two aspects. First, the new algorithm can learn and
  reuse D-sequents. Second, it proves clauses redundant one by one and
  thus backtracks as soon as the current target clause is proved
  redundant in the current subspace. This makes the new PQE algorithm
  similar to a SAT-solver that backtracks as soon as just \ti{one}
  clause is falsified.  We show experimentally that the new PQE
  algorithm outperforms its predecessor.
\end{abstract}

\section{Introduction}
Many verification problems reduce to Quantifier Elimination
(\tb{QE}). So, any progress in QE is of great importance. In this
paper, we consider propositional CNF formulas with existential
quantifiers.  Given formula \prob{X}{F(X,Y)} where $X$ and $Y$ are
sets of variables, the QE problem is to find a quantifier-free formula
$F^*(Y)$ such that $F^* \equiv \prob{X}{F}$. Building a practical QE
algorithm is a tall order.  In addition to the sheer complexity of QE,
a major obstacle here is that the size of formula $F^*(Y)$ can be
prohibitively large.

There are at least two ways of making QE easier to solve. First, one
can consider only instances of QE where $|Y|$ is small, which limits
the size of $F^*$. In particular, if $|Y|=0$, QE reduces to the
satisfiability problem (SAT). This line of research featuring very
efficient methods of model checking based on SAT~\cite{bmc,ken03,ic3}
has gained great popularity. Another way to address the complexity of
QE suggested in~\cite{hvc-14} is to perform \tb{partial QE}
(\tb{PQE}). Given formula \prob{X}{F_1(X,Y) \wedge F_2(X,Y)}, the PQE
problem is to find a quantifier-free formula $F^*_1(Y)$ such that
\pqe{X}{F^*_1}{F_1}{F_2}. We will say that formula $F^*_1$ is obtained
by \ti{taking} $F_1$ \ti{out of the scope of quantifiers}.

The appeal of PQE is threefold. First, intuitively, PQE should be much
simpler than QE if $F_1$ is much smaller than $F_2$. Second, PQE can
perform SAT. So one can view a PQE-solver as a SAT-solver with extra
semantic power due to using quantifiers.  Third, in addition to SAT,
many verification problems, reduce to PQE (see
Section~\ref{sec:appl_pqe}).  For instance, an equivalence checker
based on PQE~\cite{fmcad16} enables construction of short resolution
proofs of equivalence for a very broad class of structurally similar
circuits. These proofs are based on the notion of clause
redundancy\footnote{A clause is a disjunction of literals. So a CNF formula $F$ is a
conjunction of clauses: $C_1 \wedge \dots \wedge C_k$. We also
consider $F$ as the \ti{set} of clauses \s{C_1,\dots,C_k}. Clause $C$
is redundant in \prob{X}{F} if
$\prob{X}{F} \equiv \prob{X}{F \setminus \s{C}}$.
} in a \ti{quantified} formula
and thus cannot be generated by a traditional SAT-solver.
In~\cite{mc_no_inv}, we show that a PQE-solver can check if the
reachability diameter exceeds a specified value. So it can turn
bounded model checking~\cite{bmc} into unbounded as opposed to a pure
SAT-solver. Importantly, no generation of an \ti{inductive invariant}
is required by the method of~\cite{mc_no_inv}.

If $F^*_1(Y)$ is a solution to the PQE problem above, it is implied by
$F_1 \wedge F_2$. So $F^*_1$ can be obtained by resolving clauses of
$F_1 \wedge F_2$. However, a PQE-solver based on resolution \ti{alone}
cannot efficiently address the following ``termination
problem''. Suppose one builds $F^*_1$ incrementally, adding one clause
at a time. When can one terminate this procedure claiming that $F^*$
is a solution to the PQE problem (and so \pqe{X}{F^*_1}{F_1}{F_2})?
The inability of resolution to address the termination problem stems
from its ``asymmetry'' in treating satisfiable and unsatisfiable
formulas. To prove formula $G$ unsatisfiable, one just needs to add to
$G$ new resolvents until an empty clause is derived. However, proving
$G$ \ti{satisfiable} requires reaching a saturation point where every
new resolvent is implied by a clause of $G$.

In \cite{fmcad12,fmcad13} we approached the termination problem above
using the following observation.  Assume for the sake of simplicity
that every clause of $F_1$ contains at least one variable of
$X$. Then, if $F^*_1$ is a solution, $F_1$ can be dropped from $F^*_1
\wedge \prob{X}{F_1 \wedge F_2}$. Thus, $F^*_1$ becomes a solution as
soon as it makes the clauses of $F_1$ \ti{redundant}. The ability of
redundancy-based reasoning to handle the termination problem is rooted
in the fact that such reasoning enables treating satisfiable and
unsatisfiable formulas in a symmetric way (see
Section~\ref{sec:asymm}).

In \cite{hvc-14}, we introduced a PQE-solver called \apqe based on the
notion of redundancy (DS stands for ``D-Sequent''). \apqe is a
branching algorithm that, in addition to deriving new clauses and
conjoining them with $F_1 \wedge F_2$, generates dependency sequents
(\ti{D-sequents}). A D-sequent is a record saying that a clause is
redundant in a specified subspace. \apqe branches until proving
redundancy\footnote{By ''proving a clause $C$ redundant'' we mean ``\ti{making} $C$
redundant by adding new clauses (if necessary) and then proving $C$
redundant''.
} of target clauses becomes
trivial at which point so-called ``atomic'' D-sequents are
generated. The D-sequents of different branches are merged using a
resolution-like operation called \ti{join}. Upon completing the search
tree, \apqe derives D-sequents stating redundancy of the clauses of
$F_1$.

\apqe has two flaws.  First, \apqe employs ``multi-event''
backtracking.  Namely, it backtracks only when \ti{all} clauses of
$F_1$ are proved redundant in the current subspace. (This is different
from a SAT-solver that backtracks as soon as \ti{just one} clause of
the formula is falsified.)  The intuition here is that multi-event
backtracking leads to building very deep and thus very large search
trees. Second, \apqe does not reuse D-sequents derived in different
branches. The problem here is that redundancy is a \ti{structural}
rather than a \ti{semantic} property. So, a clause redundant in
formula $G'$ may not be redundant in $G''$ logically equivalent to
$G'$ (whereas a semantic property holds for \ti{all} equivalent
formulas). So, reusing a D-sequent is not as easy as reusing a clause
learned by a SAT-solver.

In this paper, we address both flaws of \apqe.  First, we present a
new PQE algorithm called \Apqe that employs \ti{single-event
  backtracking}. At any given moment, \Apqe proves redundancy of only
one clause. Once this goal is achieved, it picks a new clause to prove
redundant.  Second, \Apqe uses a new definition of D-sequents
introduced in~\cite{qe_learn}. This definition facilitates safe
reusing of D-sequents. We show experimentally that \Apqe is
significantly faster than \apqe.

This main body of the paper\footnote{Some additional information is
  provided in Appendices.} is structured as follows. Basic definitions
are given in Section~\ref{sec:basic}.  Section~\ref{sec:appl_pqe}
lists some verification problems that reduce to PQE. In
Section~\ref{sec:simp_exmp}, we give a simple example of using the
machinery of D-sequents in PQE. Section~\ref{sec:asymm} explains how
to use the notion of redundancy to treat satisfiable and unsatisfiable
formulas ``symmetrically''. In Section~\ref{sec:bnd_pnts}, we describe
the semantics of clause redundancy in terms of boundary points.
Sections~\ref{sec:dseqs}-\ref{sec:join_upd} present the machinery of
D-sequents. Section~\ref{sec:exper} provides experimental
results. Some background is given in Section~\ref{sec:bg}. Finally, in
Sections~\ref{sec:concl} and~\ref{sec:fut_dir}, we make conclusions
and describe directions for future research.

\section{Basic Definitions}
\label{sec:basic}

In this paper, we consider only propositional CNF formulas. In the
sequel, when we say ``formula'' without mentioning quantifiers we mean
a \ti{quantifier-free CNF} formula.
%
%
\begin{definition}
 \label{def:ecnf}
  Let $F$ be a CNF formula and $X$ be a subset of variables of $F$. We
  will refer to \prob{X}{F} as an \bm{\ecnf} \tb{formula}.
\end{definition}
%
%
\begin{definition}
  \label{def:vars}
Let $F$ be a CNF formula. \bm{\V{F}} denotes the set of variables of
$F$ and \bm{\V{\prob{X}{F}}} denotes $\V{F} \setminus X$.
\end{definition}

%
%
\begin{definition}
Let $V$ be a set of variables. An \tb{assignment} \pnt{q} to $V$ is a
mapping $V'~\rightarrow \s{0,1}$ where $V' \subseteq V$. We will
denote the set of variables assigned in \pnt{q} as \bm{\Va{q}}. We
will denote as \bm{\pnt{q} \subseteq \pnt{r}} the fact that a) $\Va{q}
\subseteq \Va{r}$ and b) every variable of \Va{q} has the same value in
\pnt{q} and \pnt{r}.
\end{definition}
%
%
\begin{definition}
\label{def:cofactor}
Let $C$ be a clause, $H$ be a formula that may have quantifiers, and
\pnt{q} be an assignment. \bm{\cof{C}{q} \equiv 1} if $C$ is satisfied
by \pnt{q}; otherwise it is the clause obtained from $C$ by removing
all literals falsified by \pnt{q}. \bm{\cof{H}{q}} denotes the formula
obtained from $H$ by replacing every clause $C$ with \cof{C}{q}.
\end{definition}
%
%
\begin{definition}
\label{def:formula-equiv}
Let $G, H$ be formulas that may have quantifiers. We say that $G, H$
are \tb{equivalent}, written \bm{G \equiv H}, if for all assignments
\pnt{q} where $\Va{q} \supseteq$ $(\V{G} \cup \V{H})$, we have
$\cof{G}{q} = \cof{H}{q}$.
\end{definition}
%
%
\begin{definition}
 \label{def:qe_prob}
The \tb{Quantifier Elimination (QE)} problem for formula
\prob{X}{F(X,Y)} is to find a formula $F^*(Y)$ such that \bm{F^*
  \equiv \prob{X}{F}}.
\end{definition}

%
%
\begin{definition}
 \label{def:pqe_prob}
 The \tb{Partial QE} (\tb{PQE}) problem of taking $F_1$ out of the
 scope of quantifiers in \prob{X}{F_1(X,Y) \wedge F_2(X,Y)} is to find
 formula $F^*_1(Y)$ such that \bm{F^*_1 \wedge \prob{X}{F_2} \equiv
   \prob{X}{F_1 \wedge F_2}}.
\end{definition}

%
%
\begin{remark}
\label{rem:XYsets}
From now on, we will use $X$ and $Y$ to denote sets of quantified and
non-quantified variables respectively. We will assume that variables
denoted by $x_i$ and $y_i$ are in $X$ and $Y$ respectively. Using
$X,Y$ in a quantifier-free formula implies that in the context of
QE/PQE, $X$ and $Y$ specify the quantified and non-quantified
variables respectively.
\end{remark}
%
%
\begin{definition}
  \label{def:Xcls}
Let \prob{X}{F(X,Y)} be an \ecnf formula.  A clause $C$ of $F$ is
called an \bm{X}\tb{-clause} if \V{C} $\cap~X~\neq~\emptyset$.
\end{definition}
%
%
\begin{definition}
\label{def:red_cls}
Let $F$ be a CNF formula and $G \subseteq F$ and $G \neq
\emptyset$. The clauses of $G$ are \textbf{redundant in} \bm{F} if $F
\equiv (F \setminus G)$.  The clauses of $G$ are \textbf{redundant in}
\bm{\prob{X}{F}} if $\prob{X}{F} \equiv \prob{X}{F \setminus G}$. Note
that $F \equiv (F \setminus G)$ implies $\prob{X}{F} \equiv \prob{X}{F
  \setminus G}$ but the opposite is not true.
\end{definition}

\section{Some applications of PQE}
\label{sec:appl_pqe}
In this section, we justify our interest in PQE by listing some
verification problems that can be solved by a PQE algorithm.

%
%
\subsection{Circuit-SAT}
\label{ssec:circ_sat}
Let $N(X,Y,Z)$ be a combinational circuit where $X$,$Y$,$Z$ are sets
of input, internal and output variables respectively. Let \pnt{z} be
an assignment\footnote{In this section, when we say ``an assignment
  \pnt{v} to a set of variables $V$'' we mean a full assignment (i.e.
  every variable of $V$ is assigned in \pnt{v}).} to $Z$. Consider the
problem\footnote{This problem reduces to Circuit-SAT if $Z=\s{z}$,
  $\pnt{z}=(z=1)$ and it suffices to produce just one input (if any)
  for which $N$ outputs \pnt{z}.} of finding inputs (i.e. assignments
to $X$) for which $N$ produces output \pnt{z}.
Let $F(X,Y,Z)$ be a formula specifying $N$.  Let \Sub{C}{z} denote the
longest clause falsified by \pnt{z}.  The problem above reduces to
taking \Sub{C}{z} out of the scope of quantifiers in
\prob{W}{\Sub{C}{z} \wedge F} where $W = Y \cup Z$ (see
~\cite{south_korea}).  Namely, one needs to find a formula $G(X)$ such
that $\prob{W}{\Sub{C}{z} \wedge F} \equiv$ $G \wedge
\prob{W}{F}$. Every input \pnt{x} falsifying $G$ produces the output
\pnt{z}.
%
%
\subsection{General SAT}
Let $F(X)$ be a formula to be checked for satisfiability and \pnt{x}
be an assignment to $X$. Let $F_1$ and $F_2$ denote the clauses of $F$
satisfied and falsified by \pnt{x} respectively. Then checking the
satisfiability of $F$ reduces to taking $F_1$ out of the scope of
quantifiers in \prob{X}{F_1 \wedge F_2} (see \cite{hvc-14}). That is
one just needs to find $F^*_1$ such that $\prob{X}{F_1 \wedge F_2}
\equiv F^*_1 \wedge \prob{X}{F_2}$. Since all variables of $F$ are
quantified, $F^*_1$ is a constant. If $F^*_1 = \mi{false}$, $F$ is
unsatisfiable because \prob{X}{F} $\equiv$ \prob{X}{F_1 \wedge F_2}
$\equiv \mi{false}$. If $F^*_1=\mi{true}$, $F$ is satisfiable (because
$F_2$ is satisfied by \pnt{x}).
%
%
\subsection{Interpolation}
Let $I(Y)$ be an interpolant for formulas $A(X,Y)$ and
$\overline{B}(Y,Z)$ i.e $A \Rightarrow I \Rightarrow
\overline{B}$. Let formula $A^*(Y)$ be obtained by taking $A$ out of
the scope of quantifiers in \prob{W}{A \wedge B} where $W = X \cup Z$.
That is $\prob{W}{A \wedge B} \equiv A^*\wedge \prob{W}{B}$. Assume
also that $A \Rightarrow A^*$. Then $A \Rightarrow A^* \Rightarrow
\overline{B}$ and $A^*$ is an interpolant~\cite{tech_rep_ec_lor}. So,
one can view interpolation as a special case of PQE.
%
%
\subsection{Equivalence checking}
\label{ssec:pqe_ec}
Let $N'(X',Y',z')$ and $N''(X'',Y'',z'')$ be single-output
combinational circuits to be checked for equivalence.  Here $X',Y'$
are sets of input and internal variables and $z'$ is the output
variable of $N'$. Sets $X''$ and $Y''$ and variable $z''$ have the
same meaning for $N''$. Let $\mi{EQ}(X',X'')$ specify the predicate
such that $\mi{EQ}(\pnt{x'},\pnt{x''})$ iff $\pnt{x'} = \pnt{x''}$.
Let formulas $F'(X',Y',z')$ and $F''(X'',Y'',z'')$ specify circuits
$N'$ and $N''$ respectively.

The equivalence of $N'$ and $N''$ can be checked by taking $\mi{EQ}$
from the scope of quantifiers in \prob{W}{\mi{EQ} \wedge F' \wedge
  F''} where $W = X' \cup X'' \cup Y' \cup Y''$
(see~\cite{fmcad16}). Let $h(z',z'')$ be a formula such that
\prob{W}{\mi{EQ} \wedge F' \wedge F''} $\equiv$ $h \wedge \prob{W}{F'
  \wedge F''}$.  If $h(z',z'')$ specifies $z' \equiv z''$, then $N'$
and $N''$ are equivalent. Otherwise, $N'$ and $N''$ are inequivalent
unless they implement identical constants. (This possibility can be
ruled out by a few easy SAT-checks.) 
%
%
\subsection{Model checking}
Let formulas $T(S,S')$ and $I(S)$ specify the transition relation and
initial states of a system \ks respectively. Here $S$ and $S'$ are
sets of variables specifying the present and next states
respectively. Let \di{I,T} denote the \ti{reachability diameter} of
\ks (i.e. every state of \ks is reachable in at most \di{I,T}
transitions).

Given a number $n$, one can use a PQE solver to check if $n \ge
\di{I,T}$ as follows\,\cite{mc_no_inv}.  Let $I_1=I(S_1)$ and $W_n=S_0
\cup \dots \cup S_n$ and $G_{0,n}=T_{0,1} \wedge \dots \wedge
T_{n-1,n}$ and $T_{i,i+1}=T(S_i,S_{i+1})$. Testing if $n \ge \di{I,T}$
reduces to checking if $I_1$ is redundant in \prob{W_n}{I_1 \wedge I_0
  \wedge G_{0,n+1}} i.e. whether \prob{W_n}{I_1 \wedge I_0 \wedge
  G_{0,n+1}} $\equiv$ \prob{W_n}{I_0 \wedge G_{0,n+1}}. If so, then $n
\ge \di{I,T}$. Then, to prove a safety property $P(S)$, it suffices to
run BMC~\cite{bmc} to show that no counterexample of length $n$ or
less exists.

\section{A Simple Example}
\label{sec:simp_exmp}
In this section, we present a simple example of performing PQE by
deriving D-sequents.  A D-sequent of~\cite{fmcad13} is a
record \olds{X}{F}{q}{C} stating redundancy of clause $C$ in
\prob{X}{F} in subspace \pnt{q} (where \pnt{q} is an assignment to
variables of $F$).
Let \prob{X}{C_1 \wedge G} be a formula where $X = \s{x_1,x_2}$, $C_1
= \overline{x}_1 \vee x_2$, $G = C_2 \wedge C_3$, $C_2 = y \vee x_1$,
$C_3 = y \vee \overline{x}_2$.  Consider the PQE problem of taking
$C_1$ out of the scope of quantifiers. Below we solve this problem by
proving $C_1$ redundant.

In subspace $y\!=\!0$, clauses $C_2,C_3$ are \tb{unit} (i.e.  one
literal is unassigned, the rest are falsified).  After assigning
\mbox{$x_1\!=\!1$}, $x_2\!=\!0$ to satisfy $C_2,C_3$, the clause $C_1$ is
falsified. Using the standard conflict analysis~\cite{grasp} one
derives a conflict clause $C_4 = y$. Adding $C_4$ to $C_1 \wedge G$
makes $C_1$ redundant in subspace $y=0$.  So the D-sequent $S'$ equal
to \olds{X}{F}{q'}{C_1} holds where $F = C_1 \wedge G \wedge C_4$ and
$\pnt{q'}=(y=0)$.

In subspace $y=1$, the clause $C_1$ is ``blocked'' at $x_1$. That is
no clause of $F$ is resolvable with $C_1$ on $x_1$ in subspace $y=1$
because $C_2$ is satisfied by $y=1$ (see
Subsection~\ref{ssec:third_kind}). So $C_1$ is redundant in
formula \prob{X}{F} and the D-sequent $S''$ equal
to \olds{X}{F}{q''}{C_1} holds where $\pnt{q''} = (y=1)$. D-sequents
$S'$ and $S''$ are examples of so-called atomic D-sequents. They are
derived when proving clause redundancy is trivial (see
Section~\ref{sec:atom_dseqs}). One can produce a new D-sequent
\olds{X}{F}{q}{C_1} where $\pnt{q} = \emptyset$ by ``joining'' $S'$
and $S''$ at $y$ (see Subsection~\ref{ssec:join}). This D-sequent
states the \ti{unconditional} redundancy of $C_1$ in \prob{X}{F}. So,
$C_4 \wedge \prob{X}{G} \equiv$ \prob{X}{C_1 \wedge G \wedge
C_4}. Since $C_1 \wedge G$ implies $C_4$, then
$C_4 \wedge \prob{X}{G} \equiv$ \prob{X}{C_1 \wedge G}. So $C_4$ is a
solution to our PQE problem.

\section{Redundancy And SAT/UNSAT Symmetry}
\label{sec:asymm}
As mentioned earlier, there is an obvious asymmetry in how pure
resolution treats satisfiable and unsatisfiable formulas. Namely,
resolution cannot efficiently solve satisfiable formulas. In the
SAT-solvers based on the DPLL procedure~\cite{dpll}, this problem is
addressed by building a search tree where each branch corresponds to an
assignment. For a satisfiable formula, the search terminates as soon
as a branch specifying a satisfying assignment is found.

Note that the DPLL procedure does not eliminate the asymmetry in
treating satisfiable and unsatisfiable formulas. It just simplifies
proving a formula satisfiable (by finding a satisfying assignment).
However, this does not work well when one has to enumerate \ti{many}
satisfying assignments. Consider, for instance, the QE problem of
finding $F^*(Y)$ logically equivalent to \prob{X}{F(X,Y)}. Suppose one
builds $F^*(Y)$ by a DPLL-like procedure. In the worst case, this
requires finding a satisfying assignment for every assignment \pnt{y}
to $Y$ for which $F$ is satisfiable.

One can use the notion of redundancy to recover the symmetry between
satisfiable and unsatisfiable formulas. Consider, for instance, the QE
problem above. Let $F$ be unsatisfiable in subspace \pnt{y}.  Then to
make the $X$-clauses of $F$ redundant in this subspace one needs to
add a clause $C(Y)$ (implied by $F$) that is falsified by \pnt{y}.  If
$F$ is satisfiable in subspace \pnt{y}, all $X$-clauses are already
redundant in this subspace and hence no clause needs to be added. So,
the only difference between SAT and UNSAT cases is that in the UNSAT
case one has to add add a clause to make target $X$-clauses redundant.

\section{Proving Clause Redundancy}
\label{sec:bnd_pnts}
Let $F(X)$ be a quantifier-free formula. Proving redundancy of a
clause $C \in F$ reduces to checking if $F \setminus \s{C}$ implies
$C$.  So, in this case, a redundancy check is straightforward and
reduces to SAT.  Now consider proving redundancy of $C \in F$ in
formula \prob{X}{F(X,Y)}. If $\V{C} \subseteq Y$, then the redundancy
check is still the same as above. However, the situation changes if
$C$ is an $X$-clause. The fact that an $X$-clause is redundant in
\prob{X}{F} does not mean that it is redundant in $F$ as well.

In~\cite{tech_rep_edpll}, to address the problem of proving clause
redundancy, we developed a machinery of boundary points.  Given a
formula \prob{X}{F} and a clause $C \in F$, a boundary point is a full
assignment (\pnt{x},\pnt{y}) to $X \cup Y$ that falsifies $C$ but
satisfies $F \setminus \s{C}$.  This boundary point is called
removable if there is a clause $B(Y)$ implied by $F$ that is falsified
by (\pnt{x},\pnt{y}). Adding $B$ to $F$ eliminates (\pnt{x},\pnt{y})
as a boundary point (because it does not satisfy $F \setminus \s{C}$
anymore). An $X$-clause $C$ is redundant in \prob{X}{F} if no
removable boundary point exists.

\section{Dependency Sequents (D-sequents)}
\label{sec:dseqs}
In~\cite{fmcad13}, we introduced a machinery of D-sequents meant for
dealing with quantified formulas. It can be viewed as an extension of
resolution that facilitates treating satisfiable and unsatisfiable
formulas in a symmetric way (see Section~\ref{sec:asymm}).  In this
section, we modify the definition of D-sequents introduced
in~\cite{fmcad13}.  In Subsection~\ref{ssec:old_dseq}, we explain the
reason for such a modification. The new definition is given in
Subsection~\ref{ssec:def_dseqs}.

%
%
\subsection{Motivating example}
\label{ssec:old_dseq}
Let formula \prob{X}{F} contain two identical $X$-clauses $C$ and
$B$. The presence of $C$ makes $B$ redundant and vice versa. So,
D-sequents \olds{X}{F}{q}{C} and \olds{X}{F}{q}{B} hold where $\pnt{q}
= \emptyset$.  Denote them as $S_C$ and $S_B$ respectively. (Here, we
use the old definition of D-sequents given in~\cite{fmcad13}.)  $S_C$
and $S_B$ state that $C$ and $B$ are redundant in \prob{X}{F}
\ti{individually}.  Using $S_C$ and $S_B$ \ti{together} (to remove
\ti{both} $B$ and $C$ from \prob{X}{F}) is incorrect because it
involves circular reasoning.

The problem here is that redundancy is a \ti{structural} property. So,
the redundancy of $B$ in \prob{X}{F} does not imply that of $B$ in
\prob{X}{F \setminus \s{C}} even though $F \equiv F \setminus
\s{C}$. The definition of a D-sequent given in~\cite{fmcad13} does not
help to address the problem above. This definition states redundancy
of a clause only with respect to formula \prob{X}{F}. (This makes it
hard to reuse D-sequents and is the reason why the PQE-solver
introduced in~\cite{fmcad13} does not reuse D-sequents). We address
this problem by adding a \ti{structural constraint} to the definition
of a D-sequent. It specifies a \ti{subset of formulas} where a
D-sequent holds and so helps to avoid using this D-sequent in
situations where it \ti{may not hold}. Adding structural constraints
to D-sequents $S_C$ and $S_B$ makes them mutually exclusive (see
Example~\ref{exmp:mut_excl} below).

%
%
\subsection{Definition of D-sequents}
\label{ssec:def_dseqs}
%
%
\begin{definition}
\label{def:dseq}
Let \prob{X}{F} be an \ecnf formula and \pnt{q}\, be an assignment to
\V{F}. Let $C$ be an $X$-clause of $F$ and $H$ be a subset of $F
\setminus \s{C}$.  A dependency sequent (\tb{D-sequent}) $S$ has the
form \Ods{X}{F}{q}{H}{C}. It states that clause \cof{C}{q} is
redundant in every formula \prob{X}{\cof{W}{q}} logically equivalent
to \prob{X}{\cof{F}{q}} where $H \cup \s{C} \subseteq W \subseteq F$.
\end{definition}
%
%
\begin{definition}
\label{def:aux_terms}
 The assignment \pnt{q} and formula $H$ above are called the
 \tb{conditional} and the \tb{structure constraint} of the D-sequent
 $S$ respectively. We will call \prob{X}{W}, where \mbox{$H \cup \s{C}
   \subseteq W \subseteq F$}, a \tb{member formula} of $S$.  We will
 say that a D-sequent $S$ specified by \Ods{X}{F}{q}{H}{C} \tb{holds}
 if it states redundancy of $C$ according to Definition~\ref{def:dseq}
 (i.e.  if $S$ is \ti{correct}). We will say that $S$ is
 \tb{applicable} to a formula \prob{X}{W} if the latter is a member
 formula of $S$. Otherwise, $S$ is called inapplicable to \prob{X}{W}.
\end{definition}

The structure constraint $H$ of Definition~\ref{def:dseq} specifies a
subset of formulas logically equivalent to \prob{X}{F} where the
clause $C$ is redundant.  From a practical point of view, the presence
of $H$ influences the order in which $X$-clauses can be proved
redundant. Proving an $X$-clause $B$ of $H$ redundant and removing it
from $F$ renders the D-sequent $S$ inapplicable to the modified
formula (i.e. \prob{X}{F \setminus \s{B}}). Thus, if one intends to
use $S$, the clause $B$ should be proved redundant \ti{after} $C$.
\begin{example}
  \label{exmp:mut_excl}
  Consider the example introduced in Subsection~\ref{ssec:old_dseq}.
  In terms of Definition~\ref{def:dseq}, the D-sequent $S_C$ looks
  like \Ods{X}{F}{q}{H_C}{C} where $\pnt{q} = \emptyset$, $H_C =
  \s{B}$ (because the presence of clause $B$ is used to prove $C$
  redundant).  Similarly, the D-sequent $S_B$ looks like
  \Ods{X}{F}{q}{H_B}{B} where $H_B = \s{C}$. D-sequents $S_C$ and
  $S_B$ are mutually exclusive: using $S_C$ to remove $C$ from $F$ as
  a redundant clause renders $S_B$ inapplicable and vice versa.
\end{example}
%
%
\begin{remark}
\label{rem:short_dseqs}
We will abbreviate D-sequent $(\prob{X}{F},\!\pnt{q},\!H)\!\rightarrow\!C$
to \ods{q}{H}{C} if \prob{X}{F} is known from the context.
\end{remark}

\section{Reusing Single And Multiple D-sequents}
In this section, we discuss conditions under which single and multiple
D-sequents can be safely reused.
%
%
\subsection{Reusing a single D-sequent}
\label{ssec:sngl_dseq_reuse}
Let $S$ be a D-sequent specified by \Ods{X}{F}{q}{H}{C}.
 We will say that $S$ is \tb{active} in subspace \pnt{r}
for formula \prob{X}{W}  if
\begin{itemize}
\item $\pnt{q} \subseteq \pnt{r}$ and
\item $S$ is applicable to \prob{X}{W} (see
  Definition~\ref{def:aux_terms})
\end{itemize}
The activation of $S$ means that it can be safely \ti{reused} (i.e.
$C$ can be dropped in the subspace \pnt{r} as
redundant\footnote{Redundancy of a clause in subspace \pnt{q} does not \ti{trivially}
imply its redundancy in subspace $\pnt{q} \subset \pnt{r}$, i.e. in a
smaller subspace (see Appendix~\ref{app:red_subsp}).
} in \prob{X}{W}).

An applicable D-sequent $S$ equal to \Ods{X}{F}{q}{H}{C} is called
\tb{unit} under assignment \pnt{r} if all values assigned in \pnt{q}
\ti{but one} are present in \pnt{r}.  Suppose, for instance,
$\pnt{q}=(y_1 = 0, x_5 = 1)$ and \pnt{r} contains $y_1=0$ but $x_5$ is
not assigned in \pnt{r}. Then $S$ is unit. Adding the assignment
\mbox{$x_5=1$} to \pnt{r}, activates $S$, which indicates that $C$ is
redundant in the subspace $\pnt{r} \cup \s{x_5=1}$.  So, a unit
D-sequent can be used like a unit clause in Boolean Constraint
Propagation (BCP) of a SAT-solver.  Namely, one can use $S$ to derive
the ``deactivating'' assignment $x_5 = 0$ as a direction to a subspace
where $C$ is not proved redundant yet.

%
%
\subsection{Reusing a set of D-sequents}
\label{ssec:mult_dseq_reuse}
In Example~\ref{exmp:mut_excl}, we described D-sequents that
\ti{cannot} be active together. Below, we introduce a condition under
which a set of D-sequents \ti{can} be active together.
%
%
\begin{definition}
  Assignments \pnt{q'} and \pnt{q''} are called \tb{compatible} if
  every variable of $\Va{q'} \cap \Va{q''}$ is assigned the same value
  in \pnt{q'} and \pnt{q''}.
\end{definition}
%
%
\begin{definition}
  \label{def:cons_dseqs}
  Let \prob{X}{F} be an \ecnf formula. Let $S_1,\dots,S_k$ be
  D-sequents specified by
  $(\ppnt{q}{1},H_1)\!\rightarrow\!C_1$,$\dots$,
  $(\ppnt{q}{k},\!H_k)\!\rightarrow\!C_k$ respectively. They are
  called \tb{consistent} if a) every pair \ppnt{q}{i},\ppnt{q}{j},
  $1\!  \leq i,j\!\leq k$ is compatible and b) there is an order $\pi$
  on \s{1,\!\dots,\!k} such that \prob{X}{F\!\setminus
    \s{C_{\pi(1)},\!\dots,\!C_{\pi(m-1)}}} obtained after using
  D-sequents $S_{\pi(1)},\!\dots,\!S_{\pi(m-1)}$ is a member formula
  of $S_{\pi(m)}$, $\forall{m}\!\in\!\s{2,\!\dots\!,k}$.
\end{definition}

The item b) above means that $S_1,\dots,S_k$ can be active together if
there is an order $\pi$ following which one guarantees the
applicability of \ti{every D-sequent}. (The D-sequents $S_C$ and $S_B$
of Example~\ref{exmp:mut_excl} are \ti{inconsistent} because such an
order does not exist. Applying one D-sequent makes the other
inapplicable.)  Definition~\ref{def:cons_dseqs} specifies a
\ti{sufficient} condition for a set of D-sequents to be active
together in a subspace \pnt{r} where $\ppnt{q}{i} \subseteq \pnt{r}$,
$1 \leq i \leq k$. If this condition is met, $C_1,\dots,C_k$ can be
safely removed from \prob{X}{F} in the subspace \pnt{r}
(see~\cite{qe_learn}).

\section{Atomic D-sequents}
\label{sec:atom_dseqs}
In this section, we describe D-sequents called atomic. An atomic
D-sequent is generated when proving a clause redundant is
trivial~\cite{fmcad13}. We modify the definitions of~\cite{fmcad13} to
accommodate the appearance of a structure constraint.

%
%
\subsection{Atomic D-sequents of the first kind}
\begin{proposition}
  \label{prop:sat_cls}
Let \prob{X}{F} be an \ecnf formula and $C\!\in\!F$ and $v \in
\V{C}$. Let $v=b$ where $b\!\in\!\s{0,1}$ satisfy $C$. Then the
D-sequent $(\pnt{q},H)\!\rightarrow\!C$ holds where
$\pnt{q}\!=\!(v=b)$ and $H\!=\!\emptyset$. We will refer to it as an
atomic D-sequent of the \tb{first kind}.
\end{proposition}

\tb{Proofs} of all propositions can be found in~\cite{qe_learn}.
Satisfying $C$ by an assignment does not require the presence of any
other clause of $F$.  Hence, the structure constraint of a D-sequent
of the first kind is an empty set of clauses.
%
%
\begin{example}
  Let \prob{X}{F} be an \ecnf formula and $C= y_1 \vee \overline{x}_5$
  be a clause of $F$.  Since $C$ is satisfied by assignments $y_1 = 1$
  and $x_5=0$, D-sequents \oDs{y_1=1}{\emptyset}{C} and
  \oDs{x_5=0}{\emptyset}{C} hold.
\end{example}
%
%
\subsection{Atomic D-sequents of the second kind}
\label{ssec:second_kind}
\begin{proposition}
  \label{prop:uns_cls}
 Let \prob{X}{F} be an \ecnf formula and \pnt{q} be an assignment to
 \V{F}.  Let $C$ and $B$ be clauses of $F$ and $C$ be an
 $X$-clause. Let \cof{C}{q} still be an $X$-clause and \cof{B}{q}
 imply \cof{C}{q} (i.e. every literal of \cof{B}{q} is in
 \cof{C}{q}). Then the D-sequent \ods{q}{H}{C} holds where $H =
 \s{B}$.  We will refer to it as an atomic D-sequent of the \tb{second
   kind}.
\end{proposition}

%
%
\begin{example}
  Let \prob{X}{F} be an \ecnf formula.  Let $B=y_1 \vee x_2$ and $C =
  x_2 \vee \overline{x}_3$ be clauses of $F$. Let $\pnt{q}=(y_1=0)$.
  Since \cof{B}{q} implies \cof{C}{q} the D-sequent \ods{q}{\s{B}}{C}
  holds.
\end{example}
%
%
\subsection{Atomic D-sequents of the third kind}
\label{ssec:third_kind}
%
%
\begin{definition}
\label{def:resol}
Let clauses $C'$,$C''$ have opposite literals of exactly one variable
$v \in \V{C'} \cap \V{C''}$.  The clause $C$ having all literals of
$C',C''$ but those of $v$ is called the \tb{resolvent} of
$C'$,$C''$ on $v$. The clause $C$ is said to be obtained by
\tb{resolution} on $v$.  Clauses $C'$,$C''$ are called \tb{resolvable}
on~$v$.
\end{definition}
%
%
\begin{definition}
\label{def:blk_cls}
 A clause $C$ of a CNF formula $F$ is called \tb{blocked} at variable
 $v$, if no clause of $F$ is resolvable with $C$ on $v$.  The notion
 of blocked clauses was introduced in~\cite{blocked_clause}.
\end{definition}

If a clause $C$ of an \ecnf formula is blocked with respect to a
quantified variable in a subspace, it is redundant in this
subspace. This fact is used by the proposition below.

%
%
\begin{proposition}
  \label{prop:dseq_third_kind}
Let \prob{X}{F} be an \ecnf formula. Let $C$ be an $X$-clause of $F$
and $v \in(\V{C} \cap X)$. Let $C_1,\dots,C_k$ be the clauses of $F$
resolvable with $C$ on variable $v$. Let
\noDs{\ppnt{q}{1}}{H_1}{C_1},$\dots$,\noDs{\ppnt{q}{k}}{H_k}{C_k} be
consistent D-sequents (see Definition~\ref{def:cons_dseqs}).  Then the
D-sequent \noDs{q}{H}{C} holds where
\pnt{q}=$\bigcup\limits_{i=1}^{i=k}\ppnt{q}{i}$ and $H =
\bigcup\limits_{i=1}^{i=k}H_i$. We will refer to it as an atomic
D-sequent of the \tb{third kind}.
\end{proposition}

%
%
\begin{example}
  Let \prob{X}{F} be an \ecnf formula. Let $C_1,C_2,C_3$ be the only
  clauses of $F$ with variable $x_1 \in X$ where $C_1 = x_1 \vee x_2
  $, $C_2 = y_1 \vee \overline{x}_1$, $C_3 = y_2 \vee
  \overline{x}_1$. Since $y_1=1$ satisfies $C_2$, the D-sequent
  \oDs{y_1=1}{\emptyset}{C_2} holds. Suppose that the D-sequent
  \oDs{x_2=1}{\s{C_4}}{C_3} holds where $C_4~\in~F$.  Note that the
  two D-sequents above are consistent. So, from
  Proposition~\ref{prop:dseq_third_kind} it follows that the D-sequent
  \noDs{\pnt{q}}{\s{C_4}}{C_1} holds where $\pnt{q} =
  (y_1=1,x_2=1)$. The clause $C_1$ is redundant in the subspace
  \pnt{q} because it is blocked at $x_1$ in this subspace.
\end{example}

\section{Joining And Updating D-sequents}
\label{sec:join_upd}
In this section, we recall two methods for producing a new D-sequent
from existing ones~\cite{fmcad13}.  In Subsection~\ref{ssec:join}, we
present a resolution-like operation called \ti{join} that produces a
new D-sequent from two parent
D-sequents. Subsection~\ref{ssec:new_impl} describes how a D-sequent
is recomputed after adding implications.  We modify the description of
these methods given in~\cite{fmcad13} to accommodate the appearance of
a structure constraint.
In Appendix~\ref{app:gen_new_dseqs}, we describe one more way to
produce a new D-sequent that was not described in~\cite{fmcad13}.
%
%
\subsection{Join operation}
\label{ssec:join}

%
%
\begin{definition}
\label{def:res_part_assgns}
Let \pnt{q'} and \pnt{q''} be assignments in which exactly one
variable $v \in \Va{q'} \cap \Va{q''}$ is assigned different values.
The assignment \pnt{q} consisting of all the assignments of \pnt{q'}
and \pnt{q''} but those to $v$ is called the \ti{resolvent} of
\pnt{q'},\,\pnt{q''} on $v$.  Assignments \pnt{q'},\,\pnt{q''} are
called \ti{resolvable} on $v$.
\end{definition}
%
%
\begin{proposition}
\label{prop:join_rule}
Let \prob{X}{F} be an \ecnf formula. Let D-sequents
$(\pnt{q'}\!,\!H')\!\rightarrow\!C$ and
$(\pnt{q''}\!,\!H'')\!\rightarrow\!C$ hold. Let \pnt{q'}, \pnt{q''} be
resolvable on $v$ and \pnt{q}\, be the resolvent. Then the D-sequent
\ods{q}{H}{C} holds where $H = H' \cup H''$.
\end{proposition}

%
%
\begin{definition}
\label{def:join_rule}
We will say that the D-sequent \ods{q}{H}{C} of
Proposition~\ref{prop:join_rule} is produced by \tb{joining
  D-sequents} \ods{q'}{H'}{C} and \ods{q''}{H''}{C} \tb{at variable}
\bm{v}.
\end{definition}
%
%
\begin{example}
  Let \prob{X}{F(X,Y)} be an \ecnf formula. Let $C_1,C_2,C_3$ be
  clauses of $F$ and $C_1$ be an $X$-clause.  Let
  $(\pnt{q'},\!H')\!\rightarrow\!C_1$,
  $(\pnt{q''},\!H'')\!\rightarrow\!C_1$ be D-sequents were
  $\pnt{q'}\!=\!(y_1\!=\!0,x_1=0)$,
  $\pnt{q''}\!=\!(y_1\!=\!1,x_2\!=\!1)$, $H'=\s{C_2}$,
  $H''=\s{C_3}$. By joining them at $y_1$, one produces the D-sequent
  \ods{q}{H}{C_1} where $\pnt{q}\!=\!(x_1=0,x_2=1)$ and $H =
  \s{C_2,C_3}$.
\end{example}

%
%
\subsection{Updating D-sequents after adding an implication}
\label{ssec:new_impl}

%
%
As we mentioned earlier, proving redundancy of $X$-clauses of
\prob{X}{F} requires adding new clauses implied by $F$. The
proposition below shows that the D-sequents learned for \prob{X}{F}
before can be trivially updated.
\begin{proposition}
  \label{prop:dseq_add_cls}
 Let D-sequent \Ods{X}{F}{q}{H}{C} hold and $R$ be a formula implied
 by $F$. Then the D-sequent \Ods{X}{F \wedge R}{q}{H}{C} holds too.
\end{proposition}

\section{Introducing \Apqe}
\label{sec:Apqe}

In this section, we describe a PQE-algorithm called \Apqe. As we
mentioned earlier, in contrast to \apqe of~\cite{hvc-14},\,\Apqe uses
single-event backtracking. Namely, \Apqe proves redundancy of
$X$-clauses one by one and backtracks as soon as the current target
$X$-clause is proved redundant in the current subspace. Besides, due
to introduction of structure constraints, it is safe for \Apqe to
reuse D-sequents. A proof of correctness of \Apqe is given in
Appendix~\ref{app:sound_compl}.

%
%
\vspace{-3pt}
\subsection{Main loop of \Apqe}
%
%
\setlength{\intextsep}{4pt}
\setlength{\textfloatsep}{10pt}
\begin{wrapfigure}{L}{1.4in}
\small
\begin{tabbing}
a\=bb\= cc\= ddddddd\= \kill
$\Apqe(F_1,F_2 \del X)$\{\\
\tb{\scriptsize{1}}\>\,\,$\mi{Ds} := \emptyset$ \\
\tb{\scriptsize{2}}\> \,\,while $(\mi{true})$ \{ \\
\tb{\scriptsize{3}}\Tt  $C:= \mi{PickXcls}(F_1)$   \\
\tb{\scriptsize{4}}\Tt  if $(C = \mi{nil})$ return($F_1$)\\
\tb{\scriptsize{5}}\Tt   $\mi{PrvRed}(F_1,\!F_2,\!\mi{Ds}\!\del\!C,X)$\\
\tb{\scriptsize{6}}\Tt   $F_1 := F_1 \setminus \s{C}$ \}\}\\
\end{tabbing} 
\vspace{-20pt}
\caption{Main loop}
\label{fig:mloop}
\end{wrapfigure}

The main loop of \Apqe is shown in Fig.~\ref{fig:mloop}. \Apqe accepts
formulas $F_1(X,Y),F_2(X,Y)$ and set $X$ and outputs
formula $F^*_1(Y)$ such that $\prob{X}{F_1 \wedge F_2} \equiv F^*_1
\wedge \prob{X}{F_2}$.  We use symbol '$\del$' to separate
in/out-parameters and in-parameters. For instance, the line
$\Apqe(F_1,F_2 \del X)$ means that formulas $F_1,F_2$ \ti{change} by
\Apqe (via adding/removing clauses) whereas $X$ does not.

\Apqe first initializes the set \ti{Ds} of learned D-sequents. It
starts an iteration of the loop with picking an $X$-clause $C \in F_1$
(line 3). If every clause of $F_1$ contains only variables of $Y$,
then $F_1$ is a solution $F^*_1(Y)$ to the PQE problem above (line
4). Otherwise, \Apqe invokes a procedure called \prd to prove $C$
redundant. This may require adding new clauses to $F_1$ and $F_2$.  In
particular, \prd may add to $F_1$ new $X$-clauses to be proved
redundant in later iterations of the loop. Finally, \Apqe removes $C$
from $F_1$ (line 6).
%
%
\subsection{Description of \prd procedure}
\label{ssec:prv_red}
%
%
\setlength{\intextsep}{4pt}
\setlength{\textfloatsep}{10pt}
\begin{wrapfigure}{L}{1.6in}
\small
\vspace{-5pt}
\begin{tabbing}
aa\=bb\= cc\= dd\= \kill
// $\eta$ denotes $(Q,\pnt{a},T,\Sub{C}{trg})$ \\
// $\xi$ denotes $(F_1,F_2,\mi{Ds},X)$ \\
// $\phi$ denotes $(F_1,F_2,\mi{Ds},\pnt{a},T)$ \\
// \\
$\mi{PrvRed}(F_1,F_2,\mi{Ds} \del \Sub{C}{pr},X)$\{\\
\tb{\scriptsize{1}}\> $T := \mi{InitStack}(\Sub{C}{pr})$  \\
\tb{\scriptsize{2}}\> $\pnt{a} := \emptyset$; $Q := \emptyset$ \\
\tb{\scriptsize{3}}\> $\Sub{C}{trg} := \Sub{C}{pr}$ \\
$~~~-----$\\
\tb{\scriptsize{4}}\> while ($\mi{true}$) \{\\
\tb{\scriptsize{5}}\Tt if $(Q = \emptyset)$ \{ \\
\tb{\scriptsize{6}}\ttt  $(v,b)\!:=\!Assgn\!(F_1,\!F_2,\!X)$\\
\tb{\scriptsize{7}}\ttt  $\mi{UpdQueue}(Q \del v,b)$ \} \\
\tb{\scriptsize{8}}\Tt $(\mi{ans},C',S')\!:=\!\mi{BCP}(\eta \del \xi)$\\
\tb{\scriptsize{9}}\Tt if $(\mi{ans} = \mi{NoBcktr})$ \\
\tb{\scriptsize{10}}\ttt  continue\\
$~~~-----$\\
\tb{\scriptsize{11}}\Tt $(S,C)\!:=\!\mi{Lrn}(\!\mi{ans},\!\phi,\!C',\!S'\!)$\\
\tb{\scriptsize{12}}\Tt $\mi{Store}(F_1,F_2,\mi{Ds} \del S,C)$ \\
\tb{\scriptsize{13}}\Tt if $(\Sub{C}{trg} = \Sub{C}{pr})$ \{\\
\tb{\scriptsize{14}}\ttt  $\mi{RegBcktr}(\pnt{a},T \del S,C)$ \\
\tb{\scriptsize{15}}\ttt if $(\mi{Cond}(S)=\emptyset)$ return \\
\tb{\scriptsize{16}}\ttt $\mi{UpdQueue}(Q\del \pnt{a},S,C)$  \\
\tb{\scriptsize{17}}\ttt continue \} \\
$~~~-----$\\
\tb{\scriptsize{18}*}\Tt $\mi{SpecBcktr}(\pnt{a},T \del S,C)$ \\
\tb{\scriptsize{19}*}\Tt if $(\neg\mi{TrgDone}(T))$ \{ \\
\tb{\scriptsize{20}*}\ttt $\mi{UpdQueue}(\!Q\del\!\pnt{a},\!S,\!C)$  \\
\tb{\scriptsize{21}*}\ttt continue \} \\
\tb{\scriptsize{22}*}\Tt $(\!\Sub{C}{trg},S) := \mi{NewTrg}(\phi\del)$ \\
\tb{\scriptsize{23}*}\Tt if $(\Sub{C}{trg} \neq \Sub{C}{pr})$ continue\\
\tb{\scriptsize{24}*}\Tt if $(\mi{Cond}(S)=\emptyset)$ return \\
\tb{\scriptsize{25}*}\Tt $\mi{UpdQueue}(Q\del \pnt{a},S)$\}\}\\
\end{tabbing} 
\vspace{-15pt}
\caption{The \prd procedure}
\label{fig:prv_red}
\end{wrapfigure}
 The pseudo-code of \prd is shown in
Fig~\ref{fig:prv_red}. The objective of \prd is to prove a clause of
$F_1$ redundant. We will refer to this clause as the \ti{primary
  target} and denote it as \Sub{C}{pr}. To prove \Sub{C}{pr}
redundant, \prd, in general, needs to prove redundancy of other
$X$-clauses called \ti{secondary targets}. At any given moment, \prd
tries to prove redundancy of only one $X$-clause. If a new secondary
target is selected, the current target is pushed on a stack $T$ to be
finished later. How \Apqe manages secondary targets is described in
Section~\ref{sec:trg_change}.  (The lines of code relevant to this
part of \Apqe are marked in Fig.~\ref{fig:prv_red} and~\ref{fig:bcp}
\ti{with an asterisk}.)

First, \prd initializes its variables (lines 1-3). The stack $T$ of
the target $X$-clauses is initialized to \Sub{C}{pr}. The current
assignment \pnt{a} to $X \cup Y$ is initially empty. So is the
assignment queue $Q$. The current target clause \Sub{C}{trg} is set to
\Sub{C}{pr}. The main work is done in a while loop that is similar to
the main loop of a SAT-solver~\cite{grasp}. In particular, \prd uses
the notion of a \ti{decision level}. The latter consists of a decision
assignment and implied assignments derived by BCP. Decision level
number 0 is an exception: it consists only of implied assignments.
BCP derives implied assignments from unit clauses \ti{and} from unit
D-sequents (see Subsection \ref{ssec:sngl_dseq_reuse}).

 The operation of \prd in the while loop can be partitioned into three
 parts identified by dotted lines. The first part (lines 5-10) starts
 with checking if the assignment queue $Q$ is empty. If so, a new
 assignment $v\!=\!b$ is picked (line 6) where $v\!\in\!(X \cup Y)$
 and $b \in \s{0,1}$ and added to $Q$. \prd first
 assigns\footnote{The reason for making decision assignments on variables of $Y$ before
those of $X$ is as follows. The final goal of PQE is to derive clauses
$C(Y)$ making $F_1$ redundant in \prob{X}{F_1 \wedge F_2}. Giving
preference to variables of $Y$ simplifies generation of such clauses.
If a conflict occurs at a decision level started by an assignment to
variable $v \in Y$, one can easily derive a conflict clause depending
only on variables of $Y$.

} the variables of $Y$. So $v \in
 X$, only if all variables of $Y$ are assigned. Then \prd calls the
 \ti{BCP} procedure. If \ti{BCP} identifies a backtracking condition,
 \prd goes to the second part.  (This means that \Sub{C}{trg} is
 proved redundant in the subspace \pnt{a}. In particular, a
 backtracking condition is met if \ti{BCP} falsifies a clause $C'$ or
 activates a D-sequent $S'$ learned earlier.)  Otherwise, \prd begins
 a new iteration.

 \prd starts the second part (lines 11-17) with generating a conflict
 clause $C$ or a new D-sequent $S$ for \Sub{C}{trg} (line 11). Then
 \prd stores $S$ in \ti{Ds} (if it is worth reusing) or adds $C$ to
 $F_1 \wedge F_2$. If a clause of $F_1$ is used in generation of $C$,
 the latter is added to $F_1$. Otherwise, $C$ is added to $F_2$.  If
 the current target is \Sub{C}{pr}, one uses ``regular'' backtracking
 (lines 13-14, see Subsection~\ref{ssec:reg_bcktr}). If the
 conditional of $S$ is empty, \prd terminates (line 15). Otherwise, an
 assignment derived from $S$ or $C$ is added to \pnt{a} (line
 16). This derivation is possible because after backtracking, the
 generated conflict clause $C$ (or the D-sequent $S$) becomes
 \ti{unit}. If the assignment above is derived from $S$, \prd keeps
 $S$ until the decision level of this assignment is eliminated (even
 if $S$ is not stored in \ti{Ds}). The third part (lines 18-25) is
 described in Section~\ref{sec:trg_change}.
 %
 \subsection{BCP}
 \label{ssec:bcp}
 The main loop of \ti{BCP} consists of the three parts shown in
 Fig.~\ref{fig:bcp} by dotted lines. (Parameters $\eta$ and $\xi$ are
 defined in Fig.~\ref{fig:prv_red}.) In the first part (lines 2-9),
 \ti{BCP} extracts an assignment $w\!=\!b$ from the assignment queue
 $Q$ (line 2). It can be a decision assignment or one derived from a
 clause $C$ or D-sequent $S$.  Then, \ti{BCP} updates the current
 assignment \pnt{a} (line 9). Lines 3-8 are explained in
 Subsection~\ref{ssec:sec_targ}.

%
%
%
\setlength{\intextsep}{4pt}
\setlength{\textfloatsep}{10pt}
\begin{wrapfigure}{l}{1.6in}
\small
\begin{tabbing}
aa\=bb\= cc\= dd\= \kill
$\mi{BCP}(\eta \del \xi)$ \{ \\
\tb{\scriptsize{1}}\> while $(Q \neq \emptyset)$ \{ \\
\tb{\scriptsize{2}}\Tt $(w,b,C,S) := \mi{Pop}(Q\del)$  \\
\tb{\scriptsize{3}*}\Tt if $(C = \Sub{C}{trg})$ \{\\
\tb{\scriptsize{4}*}\ttt  $C':=\mi{BCP}^*(\eta \del \xi,w,b)$ \\
\tb{\scriptsize{5}*}\ttt  $\Sub{C}{trg}:= \mi{NewTrg}(T)$ \\
\tb{\scriptsize{6}*}\ttt  if ($C' \neq \mi{nil}$) \\
\tb{\scriptsize{7}*}\tttt   return($\mi{FlsCls},C',\mi{nil}$) \\
\tb{\scriptsize{8}*}\ttt  break; \} \\
\tb{\scriptsize{9}}\Tt $\mi{UpdAssgn}(\pnt{a} \del w,b,C,S)$ \\
$~~~-----$ \\
\tb{\scriptsize{10}}\Tt if $(\mi{Satisf}(\Sub{C}{trg},w,b))$ \\
\tb{\scriptsize{11}}\ttt return($\mi{SatTrg},\mi{nil},\mi{nil}$)\}   \\
\tb{\scriptsize{12}}\Tt $C'\!:=\!\mi{ChkCls}(\!Q\!\del\!F_1,\!F_2,\!w,\!b)$   \\
\tb{\scriptsize{13}}\Tt if $(C' \neq \mi{nil})$ \\
\tb{\scriptsize{14}}\ttt  return($\mi{FlsCls},C',\mi{nil}$)\\
\tb{\scriptsize{15}}\Tt $S'\!:=\!\mi{ChkDsq}(Q\!\del\!\mi{Ds},\!w,\!b)$   \\
\tb{\scriptsize{16}}\Tt if $(S'\neq \mi{nil})$ \\
\tb{\scriptsize{17}}\ttt  return($\mi{ActDseq},\mi{nil},S'$)\\
$~~~-----$ \\
\tb{\scriptsize{18}}\>\,\,if $(\mi{Blocked}(\Sub{C}{trg},\pnt{a},F_1,F_2))$ \\
\tb{\scriptsize{19}}\Tt return($\mi{BlkTrg},\mi{nil},\mi{nil}$)\}\\
\tb{\scriptsize{20}}\>\,\,return($\mi{NoBcktr},\mi{nil},\mi{nil}$)\} \\
\end{tabbing} 
\vspace{-20pt}
\caption{The \ti{BCP} procedure}
\label{fig:bcp}
\end{wrapfigure}

 In the second part (lines 10-17), \ti{BCP} first checks if the
 current target clause \Sub{C}{trg} is satisfied by $w=b$. If so,
 \ti{BCP} terminates returning the backtracking condition \ti{SatTrg}
 (line 11). Then \ti{BCP} identifies the clauses of
 \mbox{$F_1\!\wedge\!F_2$} satisfied or constrained by $w\!=\!b$
 (line~12). If a clause becomes unit, \ti{BCP} stores the assignment
 derived from this clause in $Q$. If a falsified clause $C'$ is found,
 \ti{BCP} terminates (lines 13-14). Otherwise, \ti{BCP} checks the
 \ti{applicable} D-sequents of \ti{Ds} stating the redundancy of
 \Sub{C}{trg} (line 15). If such a D-sequent became unit, the
 deactivating assignment is added to $Q$ (see
 Subsection~\ref{ssec:sngl_dseq_reuse}).  If an active D-sequent $S'$
 is found, \ti{BCP} terminates (lines 16-17).

 Finally, \ti{BCP} checks if \Sub{C}{trg} is blocked (lines 18-19). If
 not, \ti{BCP} reports that no backtracking condition is met (line
 20).
 %
 %
 \subsection{D-sequent generation}
 \label{ssec:dseq_gen}
 When \ti{BCP} reports a backtracking condition, the \lrn procedure
 (line 11 of Fig~\ref{fig:prv_red}) generates a conflict clause $C$ or
 a D-sequent $S$. \lrn generates a conflict clause when \ti{BCP}
 returns a falsified clause $C'$ and every implied assignment used by
 \lrn to construct $C$ is derived from a \ti{clause}~\cite{grasp}.
 Adding $C$ to $F_1 \wedge F_2$ makes the current target clause
 \Sub{C}{trg} redundant in subspace \pnt{a}.
 Otherwise\footnote{There is one case where \ti{Lrn} generates a D-sequent \ti{and} a
clause (see Appendix~\ref{ssec:dseq_and_clause}).
}, \lrn generates a D-sequent $S$
 for \Sub{C}{trg}. The D-sequent $S$ is built similarly to a conflict
 clause $C$. First, \lrn forms an initial D-sequent $S$ equal to
 \ods{q}{H}{\Sub{C}{trg}} (unless an existing D-sequent is activated
 by \ti{BCP}). The conditional \pnt{q} and structure constraint $H$ of
 $S$ depend on the backtracking condition returned by \ti{BCP}. If
 \pnt{q} contains assignments \ti{derived} at the current decision
 level, \lrn tries to get rid of them as it is done by a SAT-solver
 generating a conflict clause. Only instead of resolution, \lrn uses
 the join operation. Let $w=b$ be the assignment of \pnt{q} derived at
 the current decision level where $b \in \s{0,1}$. If it is derived
 from a D-sequent $S'$ equal to \ods{q'}{H'}{\Sub{C}{trg}}, \lrn joins
 $S$ and $S'$ at $w$ to produce a new D-sequent $S$. If $w=b$ is
 derived from a clause $B$, \lrn joins $S$ with the atomic D-sequent
 $S'$ of the second kind stating the redundancy of \Sub{C}{trg} when
 $B$ is falsified.  $S'$ is equal to \ods{q'}{H'}{\Sub{C}{trg}} where
 \pnt{q'} is the shortest assignment falsifying $B$ and $H'= \s{B}$.
 \lrn keeps joining D-sequents until it builds a D-sequent $S$ whose
 conditional does not contain assignments derived at the current
 decision level (but may contain the \ti{decision} assignment of this
 level). Appendix~\ref{app:dseq_gen} gives examples of D-sequents
 built by \lrn.
 
 %
 %
 \subsection{Regular backtracking}
 \label{ssec:reg_bcktr}
 If \Sub{C}{trg} is the primary target \Sub{C}{pr}, \prd calls the
 backtracking procedure \ti{RegBcktr} (line 14 of
 Fig.~\ref{fig:prv_red}).  If \lrn returns a conflict clause $C$,
 \ti{RegBcktr} backtracks to the smallest decision level where $C$ is
 still unit. So an assignment can be derived from $C$. (This is how a
 SAT-solver with conflict clause learning backtracks~\cite{grasp}.)
 Similarly, if \lrn returns a D-sequent $S$, \ti{RegBcktr} backtracks
 to the smallest decision level where $S$ is still unit. So an
 assignment can be derived from $S$.

\section{Using Secondary-Target Clauses}
\label{sec:trg_change}
The objective of \prd (see Fig.~\ref{fig:prv_red}) is to prove the
primary target clause \Sub{C}{pr} redundant. To achieve this goal,
\prd may need to prove redundancy of so-called secondary target
clauses.  In this section, we describe how this is done.
%
%
\subsection{The reason for using secondary targets}
\label{ssec:reason}
Let \prob{X}{F_1(X,Y) \wedge F_2(X,Y)} be an \ecnf formula. Assume
that \prd tries to prove redundancy of the clause $\Sub{C}{pr} \in
F_1$ where $\Sub{C}{pr} = y_1 \vee x_2$. Suppose that \pnt{a} is the
current assignment to $X \cup Y$ and $y_1$ is assigned 0 in \pnt{a}
whereas $x_2$ is not assigned yet. Since \Sub{C}{pr} is falsified in
subspace $\pnt{a} \cup \s{x_2=0}$, the assignment $x_2=1$ is derived
by BCP. However, the goal of \prd is to prove \Sub{C}{pr}
\ti{redundant} rather than satisfy $F_1 \wedge F_2$. The fact that
\Sub{C}{pr} is falsified in a subspace says nothing about its
redundancy in this subspace.

To address the problem above, \prd explores the subspace $\pnt{a} \cup
\s{x_2=1}$ to prove redundancy of the clauses of $F_1 \wedge F_2$
\ti{resolvable} with \Sub{C}{pr} on $x_2$. These clauses are called
\ti{secondary targets}.  Proving their redundancy results in proving
redundancy of \Sub{C}{pr}. If $F_1 \wedge F_2$ is \ti{unsatisfiable}
in the subspace $\pnt{a} \cup \s{x_2=1}$, \prd generates a conflict
clause that does not depend on $x_2$. Adding it to $F_1\wedge F_2$
makes \Sub{C}{pr} redundant in the subspace \pnt{a} and an atomic
D-sequent of the second kind is built. If $F_1 \wedge F_2$ is
\ti{satisfiable} in the subspace $\pnt{a} \cup \s{x_2=1}$, \prd simply
proves redundancy of the secondary targets. Then \Sub{C}{pr} is
\ti{blocked} at variable $x_2$ and an atomic D-sequent of the third
kind is generated stating the redundancy of \Sub{C}{pr} in the
subspace \pnt{a}.

The same strategy is used for \ti{every current} target clause
\Sub{C}{trg} (secondary or primary). Whenever \Sub{C}{trg} becomes
unit, \prd generates new secondary targets to be proved redundant.
These are the clauses of $F_1 \wedge F_2$ resolvable with \Sub{C}{trg}
on the variable that is currently unassigned in \Sub{C}{trg}.

%
%
\subsection{Generation of secondary targets}
\label{ssec:sec_targ}
To keep track of secondary targets \prd maintains a stack $T$ of
target levels. (Appendix~\ref{app:trg_lvls} gives an example of how
$T$ is updated.) The bottom level of $T$ consists of the primary
target clause \Sub{C}{pr}.  All other levels are meant for secondary
targets. Every such a level is specified by a pair $(C,w)$ where $C$
is either \Sub{C}{pr} or a secondary target clause and $w \in X$ is a
variable of $C$. They are called the \tb{key clause} and the \tb{key
  variable} of this level.  The secondary targets specified by this
level are the clauses of $F_1 \wedge F_2$ resolvable with $C$ on
$w$. The top level of $T$ specifies the current target clause
\Sub{C}{trg}. Namely, \Sub{C}{trg} is resolvable with the key clause
$C$ on the key variable $w$ of the top level of $T$. Once \Sub{C}{trg}
is proved redundant, another clause resolvable with $C$ on $w$ and not
proved redundant yet is chosen as the new target.

\ti{BCP} picks assignment $w=b$ derived from \Sub{C}{trg} \ti{only} if
$Q$ does not contain any other assignments (line 2 of
Fig.~\ref{fig:bcp}).  Lines (4-8) show what happens next. First, the
$\mi{BCP}^*$ procedure is called to make the assignment
$w=b$. $\mi{BCP}^*$ is similar to BCP of a SAT-solver: it derives
assignments only from clauses and returns a falsified clause $C'$ if a
conflict occurs.  The only difference is that every time $\mi{BCP}^*$
finds a unit clause $C$, a new target level of $T$ is generated. (The
reason is that every unit clause produced by $\mi{BCP}^*$ is
resolvable either with \Sub{C}{trg} or with some secondary target
generated by $\mi{BCP}^*$.)  Let $w$ be the unassigned variable of
$C$. Then this level is specified by the pair $(C,w)$.  It consists of
the clauses of $F_1 \wedge F_2$ resolvable with $C$ on $w$. On
completing $\mi{BCP}^*$, a new \Sub{C}{trg} is chosen among the
clauses of the top level of $T$ (line 5). If a conflict occurred
during $\mi{BCP}^*$, the \ti{BCP} procedure terminates (lines
6-7). Otherwise, the main loop of \ti{BCP} terminates (line 8).

%
%
\subsection{Special backtracking}
\label{ssec:spec_bcktr}
\prd uses a special backtracking procedure called \ti{SpecBcktr} (line
18 of Fig.~\ref{fig:prv_red}) if the current target is not the primary
target \Sub{C}{pr}. (An example of special backtracking is given in
Appendix~\ref{app:spec_bcktr}.) If \ti{Lrn} returns a conflict clause
$C$ (line 11), \ti{SpecBcktr} backtracks in the same manner as
\ti{RegBcktr} (line 14). Namely, it jumps to the smallest decision
level where $C$ is still unit.  The difference is that \ti{SpecBcktr}
also eliminates the target levels of $T$ that are jumped over. Namely,
if the key variable $w$ of a target level is unassigned by
\ti{SpecBcktr}, this level is eliminated. (Adding $C$ makes \ti{all}
$X$-clauses of $F_1 \wedge F_2$ redundant in the current subspace. So
proving redundancy of secondary targets with variable $w$ is not
needed anymore.)

Suppose \ti{Lrn} returns a D-sequent $S$. Since $S$ states redundancy
of \Sub{C}{trg}, the scope of \ti{SpecBcktr} is \ti{limited} to the
variables assigned \ti{after} the key variable $w$ of the top target
level of $T$ (to which \Sub{C}{trg} belongs). If some variables
assigned \ti{after} $w$ remain assigned on completing \ti{SpecBcktr},
proving \Sub{C}{trg} redundant is not over yet.  In this case, \prd
adds the assignment derived from $S$ to the queue $Q$ and starts the
next iteration of the while loop (lines 19-21 of
Fig.~\ref{fig:prv_red}).  Otherwise, \Sub{C}{trg} is proved redundant
up to the point of origin and \prd calls \ti{NewTrg} to look for a new
target (line 22). Namely, \ti{NewTrg} looks for a clause resolvable
with the key clause $C$ on the key variable $w$ of the top level of
$T$ that is not proved redundant yet.

 If \ti{NewTrg} fails to find a target in the top level of $T$, $C$ is
 blocked at $w$. Then \ti{NewTrg} generates a D-sequent for $C$ and
 deletes the top level of $T$. This entails returning the clauses of
 this level proved redundant back in $F_1 \wedge F_2$ and unassigning
 $w$. Then \ti{NewTrg} looks for a target in the \ti{new} top level of
 $T$ and so on.  If \ti{NewTrg} finds \Sub{C}{trg} that is not the
 primary target \Sub{C}{pr}, \prd starts a new iteration of the while
 loop (line 23). Otherwise, \ti{NewTrg} sets \Sub{C}{trg} to
 \Sub{C}{pr}.  It also returns a D-sequent $S$ for \Sub{C}{pr} since
 \Sub{C}{pr} is blocked. If the conditional of $S$ is empty,
 \Sub{C}{pr} is redundant unconditionally and \prd terminates (line
 24). Otherwise, the assignment derived from $S$ is added to $Q$ and a
 new iteration begins (line 25).

\section{Experimental Results}
\label{sec:exper}
In this section, we evaluate an implementation of
\Apqe. (Appendix~\ref{app:exper} provides more experimental data).
Our preliminary experiments showed that structure constraints can grow
very large, which makes storing D-sequents expensive. In
Appendix~\ref{app:str_con}, we discuss various methods of dealing with
this problem. In our experiments, we used the following idea.  One can
reduce the size of structure constraints by storing/reusing \ti{only}
D-sequents for the target clauses of $k$ bottom levels of the stack
$T$. In particular, one can safely reuse the D-sequents for the
primary target clause ($k=0$) without computing structure constraints
at all (see Appendix~\ref{app:str_con}).

For the evaluation of \Apqe, we use Circuit-SAT, the first problem
listed in Section~\ref{sec:appl_pqe}. We consider this problem in the
form repeatedly solved in IC3~\cite{ic3}: given a state \pnt{z}, find
the states from which \pnt{z} can be reached in one
transition. In~\cite{pdr}, it was suggested to look for the largest
\ti{subset} of these states forming a \ti{cube}. In this section, we
use a variation of this problem for evaluation of \Apqe. We
demonstrate that \Apqe dramatically outperforms \apqe
of~\cite{hvc-14}.

We also compare \Apqe with two SAT-based methods.  On examples with
deterministic Transition Relations (TRs), both methods are faster than
\Apqe. However, method~1 shows poorer results (in terms of cube size).
Method 2 is comparable with \Apqe in terms of cube size but is, in
general, inapplicable to \ti{non-deterministic} TRs. (A deterministic
TR is specified by a deterministic circuit $N$ i.e. an input to $N$
produces only one output. A deterministic TR can become
non-deterministic e.g.  after pre-processing~\cite{prepr} performed to
speed up the SAT-checks of IC3).  Importantly, no optimization
techniques are used in our implementation of \Apqe yet.  So its
performance can be dramatically improved.

Let $N(X,Y,Z)$ be a combinational circuit where $X$,$Y$ and $Z$ are
sets of input, internal and output variables respectively. Let \pnt{z}
be a full assignment to $Z$. The problem we consider is to find the
input assignments for which $N$ evaluates to $\pnt{z}$.  Let
\Sub{C}{\pnt{z}} be the longest clause falsified by \pnt{z}. Let
$F(X,Y,Z)$ be a CNF formula specifying $N$. Let $W$ denote $Y \cup
Z$. As we mentioned in Subsection~\ref{ssec:circ_sat}, the problem
above reduces to finding $G(X)$ such that \mbox{$G \wedge \prob{W}{F}
  \equiv$ \prob{W}{\Sub{C}{\pnt{z}} \wedge F}} i.e. to PQE. We assume
here that $N$ produces at least one output for every input. So,
$\prob{W}{F} \equiv 1$ and $G \equiv$ \prob{W}{\Sub{C}{\pnt{z}} \wedge
  F}. If $C \in G$, then $\overline{C}$ specifies a cube of input
assignments for which $N$ evaluates to \pnt{z}. So, a shorter clause
$C$ specifies a larger set of input assignments producing output
\pnt{z}.

%
%
\begin{table}[ht]
  \small
\caption{Taking \Sub{C}{\pnt{z}} out of the scope of quantifiers in
  \prob{W}{\Sub{C}{\pnt{z}} \wedge F}. The time limit is 100 seconds.}
\scriptsize
\begin{center}
\begin{tabular}{|p{42pt}|p{12pt}|p{22pt}|p{10pt}|p{14pt}|p{10pt}|p{14pt}|p{10pt}|} \hline
    name     &\#inps& \multicolumn{2}{c|}{\apqe} & \multicolumn{2}{c|}{\Apqe}& \multicolumn{2}{c|}{\Apqe}  \\
          & &\multicolumn{2}{c|}{} & \multicolumn{2}{c|}{no learning} & \multicolumn{2}{c|}{limited learning}\\ \cline{3-8}
         & & \#dseqs      & time      & \#dseqs & time   &\#dseqs & time    \\
         &     & $\times 10^3$ &(s.)  & $\times 10^3$    & (s.)  & $\times 10^3$ & (s.) \\ \hline
pdtvisheap00     & 37 & 231    & 3.4   & 0.9  & 0.03     & 0.1  & \tb{0.02}  \\\hline
texasifetch1p1   & 87 & $>$416    & $*$   & 6.8  & \tb{0.4} & 0.7 & \tb{0.4}  \\\hline
pdtpmsretherrtf  & 93 & $>$10,128 & $*$   & 17   &  1       & 0.6 & \tb{0.05} \\\hline
pdtvisblackjack2 & 109 & $>$6,857 & $*$   & 690  & 61       & 226   & \tb{18}  \\\hline
pdtvisvsar04     & 147 & 2,507 &  68   & 6.5  & 0.9      & 0.4  & \tb{0.06} \\\hline
texaspimainp01   & 253 & $>$35.6  &  $*$  &  71  & 30       & 8.7 & \tb{3.4}  \\\hline
eijkbs3330       & 286 & $>$17    & $*$   &  2.6 & 0.7      & 0.2  & \tb{0.2} \\\hline
nusmvtcasp2      & 325 & $>$3,301 & $*$   &  22  & 1.4      & 3.4 & \tb{0.2} \\\hline
pdtvissfeistel   & 429 & $>$2,305 & $*$   &  $>$44  & $*$      & 7.3 & \tb{13}  \\\hline
pdtpmsvsa16a     & 453 & $>$1,742 & $*$   &  171 & 87       & 1.1 & \tb{1} \\\hline

\end{tabular}                
\end{center}
\label{tbl:mc}
\end{table}


First, we compared the performance of \apqe~\cite{hvc-14}, \Apqe with
no learning and \Apqe with limited learning (only for primary
targets). We used the transition relation of a HWMCC-10 benchmark as
circuit $N$. For the sake of simplicity, we ignored the difference
between latched and combinational input variables of $N$. (In the
context of model checking, the literals of combinational variables are
supposed to be \ti{dropped} from $C \in G$ to make $\overline{C}$ a
cube of \ti{states}.)  In Table~\ref{tbl:mc}, we give a sample of the
set of benchmarks we tried that shows the general trend. The first
column gives the name of a benchmark. The second column shows the
number of input variables of $N$. The remaining columns give the
number of generated D-sequents (in thousands) and the run time for
each PQE procedure. Table~\ref{tbl:mc} shows that \Apqe without
learning outperforms \apqe due to generating fewer D-sequents. For the
same reason, \Apqe with learning outperforms \Apqe without learning.

%
%
\begin{table}[ht]
  \small
\caption{Comparison with SAT-based methods (deterministic circuits)}
\scriptsize
\begin{center}
  \begin{tabular}{|p{44pt}|p{12pt}|p{14pt}|p{10pt}|p{14pt}|p{10pt}|p{14pt}|p{10pt}|} \hline
    name     &\#inps& \multicolumn{2}{c|}{\ti{SAT}} & \multicolumn{2}{c|}{\ti{SAT}}& \multicolumn{2}{c|}{\Apqe}  \\
          & &\multicolumn{2}{c|}{\ti{method 1}} & \multicolumn{2}{c|}{method 2} & \multicolumn{2}{c|}{limited learning}\\ \cline{3-8}
                &       & \#len\text{-} & time & \#len\text{-}  &time &\#len\text{-} & time    \\
                &       & gth   & (s.)      &  gth   & (s.)   &  gth  & (s.) \\ \hline
visemodel       & 26    & 22  & 0.1  & \tb{6}   & 0   & \tb{6}   & 0.01 \\\hline
bobcohdoptdcd4  & 62    & 53  & 0.1  &  46      & 0.1 & \tb{42}  & 0.1  \\\hline
eijkbs3330      & 286   & 155 & 0.2  & \tb{43}  & 0.1 & \tb{43}  & 0.2 \\\hline
pdtvissfeistel  & 429   & 363 & 1.3  &  \tb{1}  & 0.3 & \tb{1}   & 4.8 \\\hline
139464p0        & 1,002 & 992 & 3    & 625      & 3.3 & \tb{572} & 23 \\\hline


\end{tabular}                
\end{center}
\vspace{-5pt}
\label{tbl:det}
\end{table}

The formula $G(X)$ above can also be found by SAT.  We tried two
SAT-based methods using Minisat~\cite{minisat} as a SAT-solver.
Method 1 (inspired by~\cite{blocking_clause}) is essentially a
\ti{universal} QE algorithm whereas method 2 is applicable only for
formulas derived from \ti{deterministic} circuits. Method 1 looks for
an assignment (\pnt{x},\pnt{y},\pnt{z}) satisfying $G \wedge F \wedge
\Sub{U}{\pnt{z}}$. (Originally, $G = \emptyset$.)  Here
\Sub{U}{\pnt{z}} is the set of unit clauses specifying \pnt{z}. Then
it builds the smallest assignment (\pnt{x'},\pnt{y},\pnt{z}) where
$\pnt{x'} \subseteq \pnt{x}$ that still satisfies $G \wedge F \wedge
\Sub{U}{\pnt{z}}$. Finally, the longest clause \Sub{C}{\pnt{x'}}
falsified by \pnt{x'} is added to $G$ and a new satisfying assignment
is generated. Method~1 terminates when $G \wedge F \wedge
\Sub{U}{\pnt{z}}$ is unsatisfiable. Method~2 follows the idea employed
in advanced implementations of IC3/PDR~\cite{pdr,uns_lift}.  Namely,
it uses the satisfying assignment (\pnt{x},\pnt{y},\pnt{z}) above to
build formula $R$ equal to $\Sub{U}{\pnt{x}} \wedge F \wedge G \wedge
\Sub{C}{\pnt{z}}$. Here \Sub{U}{\pnt{x}} is the set of unit clauses
specifying \pnt{x}. If circuit $N$ is \ti{deterministic}, $R$ is
unsatisfiable. Then one extracts the subset of \Sub{U}{\pnt{x}} used
in the proof of unsatisfiability of $R$. The clause made up of the
negated literals of this subset is added to $G$. Then a new satisfying
assignment is generated (if any).

%
%
\begin{wraptable}{L}{2in}
  \small
\caption{Non-deterministic circuits}
\vspace{-15pt}
\scriptsize
\begin{center}
  \begin{tabular}{|p{30pt}|p{14pt}|p{10pt}|p{14pt}|p{10pt}|} \hline
    name      & \multicolumn{2}{c|}{\ti{SAT}}& \multicolumn{2}{c|}{\Apqe}  \\
              & \multicolumn{2}{c|}{method 1} & \multicolumn{2}{c|}{limited learning}\\ \cline{2-5}
                  & \#len\text{-}  &time &\#len\text{-}  & time    \\
                   & gth  & (s.)   & gth     & (s.) \\ \hline
visemodel        &   14     & 0.04 &   \tb{6}      & 0.01 \\\hline
bob..ptdcd4      &   47     & 0.1  &   \tb{42}     & 0.1  \\\hline
eijkbs3330       &   75     & 0.2  &   \tb{43}     & 0.2 \\\hline
pdt..sfeistel    &   140    & 1.1  &   \tb{1}      & 13  \\\hline
139464p0         &   748    & 3.4  &   \tb{577}    & 32 \\\hline      
\end{tabular}                
\end{center}
\vspace{-5pt}
\label{tbl:nondet}
\end{wraptable}

Table~\ref{tbl:det} compares the SAT-based methods above with \Apqe
(with limited learning) on 5 formulas showing the general trend. All
three methods were run until 1,000 clauses of $G$ were generated or
the problem was finished. We computed the length of the shortest
clause generated by each method and its run time. Table~\ref{tbl:det}
shows that the SAT-based methods are faster than \Apqe whereas the
best clause generated by method 2 and \Apqe is shorter than that of
method 1. So method 2 is the winner.

In Table~\ref{tbl:nondet} we repeat the same experiment for
\ti{non-deterministic} versions of circuits from Table~\ref{tbl:det}.
To make the original circuit $N$ non-deterministic, we just dropped a
fraction of clauses in the formula $F$ representing $N$.  A
non-deterministic circuit $N$ may produce different outputs for the
same input.  In this case, the formula $R$ above can be satisfiable,
which renders method 2 \ti{inapplicable}. Table~\ref{tbl:nondet} shows
that method~1 is still faster than \Apqe but generates much longer
clauses.

\section{Some Background}
\label{sec:bg}

In this section, we give some background on learning in branching
algorithms\footnote{Information about algorithms performing
  \ti{complete} QE for propositional logic can be found
  in~\cite{bryant_bdds1,bdds_qe} (BDD based)
  and~\cite{blocking_clause,fabio,cofactoring,cav09,cav11,cmu,nik1,nik2}
  (SAT-based).} used in verification. For such algorithms, it is
important to share information obtained in different subspaces. An
important example of such sharing is the identification of isomorphic
subgraphs when constructing a BDD~\cite{bryant_bdds1}. Another example
is SAT-solving with conflict driven
learning~\cite{grasp,chaff,minisat}. The difference between learning
in BDDs/SAT-solvers and D-sequents is that the former is
semantic\footnote{A BDD of a formula is just a compact representation of its truth
table. A conflict clause $C$ is implied by the formula $F$ from which
$C$ is derived and implication is a semantic property of $F$.

} whereas the latter is structural.

The appeal of finding structural properties is that they are
\ti{formula-specific}. So using such properties can give a dramatic
performance improvement. An obvious example of a structural property
is symmetry. In~\cite{symm_brk1,mark_symm,symm_brk2}, the
permutational symmetry of a CNF formula $F$ is exploited via adding
``structural implications''. By a structural implication of $F$, we
mean a clause $C$ that, in general, is not implied by $F$ but
preserves the \ti{equisatisfiability} of $F \wedge C$ to $F$. For
instance, to keep only one satisfying assignment (if any) out of a set
of \ti{symmetric} ones, symmetry-breaking clauses are added to $F$.

ATPG
is another area where formula structure is exploited. In ATPG
methods~\cite{abram}, one reasons about a circuit in terms of
\ti{signal propagation}. In the classic paper~\cite{lar}, signal
propagation is simulated in a CNF formula $F$ generated for
identifying a circuit fault. Formula $F$ specifies the functionality
of correct and faulty circuits. Additional variables and clauses are
added to $F$ to facilitate signal reasoning. These extra clauses, like
in formulas with symmetries, are ``structural implications''.

The difference between D-sequents and traditional methods of
exploiting formula structure is twofold. First, redundancy is a
\ti{very general} structural property.  For that reason, the machinery
of D-sequents can be applied to \ti{any CNF formula} (e.g. a random
CNF formula). Second, a traditional way to take into account structure
is to add some kind of structural implications and then run a
verification engine performing \ti{semantic} derivations (e.g. a
SAT-solver). The machinery of D-sequents is different in that it
performs structural derivations (namely, proving redundancy of clauses
with quantified variables) \ti{all the way} until some semantic fact
is established e.g. $F^*_1 \wedge$ \prob{X}{F_2} $\equiv$ \prob{X}{F_1
  \wedge F_2}.

Removal of redundant clauses is used in preprocessing procedures of
QBF-algorithms and SAT-solvers~\cite{prepr,blocked_qbf}.  Redundant
clauses are also identified in the inner loop of SAT-solving
(inprocessing)~\cite{inproc}.  These procedures identify
\ti{unconditional} clause redundancies by recognizing some situations
where such redundancies can be easily proved.

\section{Conclusions}
\label{sec:concl}
We consider Partial Quantifier Elimination (PQE) on propositional CNF
formulas with existential quantifiers.  In PQE, only a (small)
subformula is taken out of the scope of quantifiers. The appeal of PQE
is that in many verification problems one can use PQE instead of
\ti{complete} QE and the former can be dramatically more efficient.
Earlier, we developed a PQE algorithm based on the notion of clause
redundancy. Since redundancy is a structural property, reusing learned
information is not trivial. In this paper, we provide some theory
addressing this problem. Besides, we introduce a new PQE algorithm
that performs single-event backtracking. This algorithm bears some
similarity to a SAT-solver and facilitates reusing learned
information. We show experimentally that the new PQE algorithm is
dramatically faster than its predecessor.  We believe that reusing
learned information is an important step in making PQE practical.

\section{Directions For Future Research}
\label{sec:fut_dir}
In our future research we are planning to focus on the following two
directions. First, although \Apqe shows an obvious improvement over
\apqe, it is too complex. So, we will try to find a simpler version of
\Apqe that preserves its good performance. Second, we want to relax
the decision making constraint (quantified variables are assigned
before unquantified). As we know from the practice of SAT-solving, a
poor choice of branching variables may lead to a significant
performance degradation.

\bibliographystyle{plain}
\bibliography{short_sat,local,l1ocal_hvc}

\begin{thebibliography}{10}

\bibitem{abram}
M.~Abramovici, M.~Breuer, and A.~Friedman.
\newblock {\em Digital Systems Testing and Testable Design}.
\newblock John Wiley \& Sons, 1994.

\bibitem{mark_symm}
F.~Aloul, K.~Sakallah, and I.~Markov.
\newblock Efficient symmetry breaking for boolean satisfiability.
\newblock {\em IEEE Transactions on Computers}, 55(5):549--558, May 2006.

\bibitem{bmc}
A.~Biere, A.~Cimatti, E.~Clarke, M.~Fujita, and Y.~Zhu.
\newblock Symbolic model checking using sat procedures instead of bdds.
\newblock In {\em DAC}, pages 317--320, 1999.

\bibitem{blocked_qbf}
A.~Biere, F.~Lonsing, and M.~Seidl.
\newblock Blocked clause elimination for qbf.
\newblock CADE-11, pages 101--115, 2011.

\bibitem{nik2}
N.~Bjorner and M.~Janota.
\newblock Playing with quantified satisfaction.
\newblock In {\em LPAR}, 2015.

\bibitem{nik1}
N.~Bjorner, M.~Janota, and W.~Klieber.
\newblock On conflicts and strategies in qbf.
\newblock In {\em LPAR}, 2015.

\bibitem{ic3}
A.~R. Bradley.
\newblock Sat-based model checking without unrolling.
\newblock In {\em VMCAI}, pages 70--87, 2011.

\bibitem{cav11}
J.~Brauer, A.~King, and J.~Kriener.
\newblock Existential quantification as incremental sat.
\newblock CAV-11, pages 191--207, 2011.

\bibitem{bryant_bdds1}
R.~Bryant.
\newblock Graph-based algorithms for {B}oolean function manipulation.
\newblock {\em IEEE Transactions on Computers}, C-35(8):677--691, August 1986.

\bibitem{bdds_qe}
P.~Chauhan, E.~Clarke, S.~Jha, J.H. Kukula, H.~Veith, and D.~Wang.
\newblock Using combinatorial optimization methods for quantification
  scheduling.
\newblock CHARME-01, pages 293--309, 2001.

\bibitem{uns_lift}
H.~Chockler, A.~Ivrii, A.~Matsliah, S.~Moran, and Z.~Nevo.
\newblock Incremental formal verification of hardware.
\newblock In {\em FMCAD-11}, pages 135--143, 2011.

\bibitem{symm_brk1}
J.~Crawford, M.~Ginsberg, E.~Luks, and A.~Roy.
\newblock Symmetry-breaking predicates for search problems.
\newblock In {\em Int. Conf. on Principles of Knowledge Representation and
  Reasoning}, pages 148--159, 1996.

\bibitem{dpll}
M.~Davis, G.~Logemann, and D.~Loveland.
\newblock A machine program for theorem proving.
\newblock {\em Communications of the ACM}, 5(7):394--397, July 1962.

\bibitem{symm_brk2}
J.~Devriendt, B.~Bogaerts, M.~Bruynooghe, and M.~Denecker.
\newblock Improved static symmetry breaking for sat.
\newblock In {\em SAT}, 2016.

\bibitem{prepr}
N.~E{\'e}n and A.~Biere.
\newblock Effective preprocessing in sat through variable and clause
  elimination.
\newblock In {\em SAT}, pages 61--75, 2005.

\bibitem{minisat}
N.~E{\'e}n and N.~S{\"o}rensson.
\newblock An extensible sat-solver.
\newblock In {\em SAT}, pages 502--518, Santa Margherita Ligure, Italy, 2003.

\bibitem{cofactoring}
M.~Ganai, A.Gupta, and P.Ashar.
\newblock Efficient sat-based unbounded symbolic model checking using circuit
  cofactoring.
\newblock ICCAD-04, pages 510--517, 2004.

\bibitem{tech_rep_ec_lor}
E.~Goldberg.
\newblock Equivalence checking by logic relaxation.
\newblock Technical Report arXiv:1511.01368 [cs.LO], 2015.

\bibitem{fmcad16}
E.~Goldberg.
\newblock Equivalence checking by logic relaxation.
\newblock In {\em FMCAD-16}, pages 49--56, 2016.

\bibitem{mc_no_inv}
E.~Goldberg.
\newblock Property checking without invariant generation.
\newblock Technical Report arXiv:1602.05829 [cs.LO], 2016.

\bibitem{qe_learn}
E.~Goldberg.
\newblock Quantifier elimination with structural learning.
\newblock Technical Report arXiv: 1810.00160 [cs.LO], 2018.

\bibitem{fmcad12}
E.~Goldberg and P.~Manolios.
\newblock Quantifier elimination by dependency sequents.
\newblock In {\em FMCAD-12}, pages 34--44, 2012.

\bibitem{tech_rep_edpll}
E.~Goldberg and P.~Manolios.
\newblock Removal of quantifiers by elimination of boundary points.
\newblock Technical Report arXiv:1204.1746 [cs.LO], 2012.

\bibitem{fmcad13}
E.~Goldberg and P.~Manolios.
\newblock Quantifier elimination via clause redundancy.
\newblock In {\em FMCAD-13}, pages 85--92, 2013.

\bibitem{hvc-14}
E.~Goldberg and P.~Manolios.
\newblock Partial quantifier elimination.
\newblock In {\em Proc. of HVC-14}, pages 148--164. Springer-Verlag, 2014.

\bibitem{south_korea}
E.~Goldberg and P.~Manolios.
\newblock Software for quantifier elimination in propositional logic.
\newblock In {\em ICMS-2014,Seoul, South Korea, August 5-9}, pages 291--294,
  2014.

\bibitem{inproc}
M.~J\"{a}rvisalo, M.~Heule, and A.~Biere.
\newblock Inprocessing rules.
\newblock IJCAR-12, pages 355--370, 2012.

\bibitem{cav09}
J.~Jiang.
\newblock Quantifier elimination via functional composition.
\newblock In {\em Proceedings of the 21st International Conference on Computer
  Aided Verification}, CAV-09, pages 383--397, 2009.

\bibitem{fabio}
H.~Jin and F.Somenzi.
\newblock Prime clauses for fast enumeration of satisfying assignments to
  boolean circuits.
\newblock DAC-05, pages 750--753, 2005.

\bibitem{cmu}
W.~Klieber, M.~Janota, J.Marques-Silva, and E.~Clarke.
\newblock Solving qbf with free variables.
\newblock In {\em CP}, pages 415--431, 2013.

\bibitem{blocked_clause}
O.~Kullmann.
\newblock New methods for 3-sat decision and worst-case analysis.
\newblock {\em Theor. Comput. Sci.}, 223(1-2):1--72, 1999.

\bibitem{lar}
T.~Larrabee.
\newblock Test pattern generation using boolean satisfiability.
\newblock {\em IEEE Trans. on Comp.-Aided Design of Integr. Circuits and
  Systems}, 11(1):4--15, Jan 1992.

\bibitem{grasp}
J.~Marques-Silva and K.~Sakallah.
\newblock Grasp -- a new search algorithm for satisfiability.
\newblock In {\em ICCAD-96}, pages 220--227, 1996.

\bibitem{blocking_clause}
K.~McMillan.
\newblock Applying sat methods in unbounded symbolic model checking.
\newblock In {\em Proc. of CAV-02}, pages 250--264. Springer-Verlag, 2002.

\bibitem{ken03}
K.~McMillan.
\newblock Interpolation and sat-based model checking.
\newblock In {\em CAV-03}, pages 1--13. Springer, 2003.

\bibitem{chaff}
M.~Moskewicz, C.~Madigan, Y.~Zhao, L.~Zhang, and S.~Malik.
\newblock Chaff: engineering an efficient sat solver.
\newblock In {\em DAC-01}, pages 530--535, New York, NY, USA, 2001.

\bibitem{pdr}
E.~Niklas, A.~Mishchenko, and R.~Brayton.
\newblock Efficient implementation of property directed reachability.
\newblock In {\em FMCAD-11}, pages 125--134, 2011.

\end{thebibliography}
\vspace{15pt}
\appendices
%
\section{Redundancy of a clause in a subspace}
\label{app:red_subsp}
In this appendix, we discuss the following problem.  Let \prob{X}{F}
be an \ecnf formula. Let $C$ be an $X$-clause of $F$ redundant in
\prob{X}{F} in subspace \pnt{q}. Let \pnt{r} be an assignment to \V{F}
where $\pnt{q} \subset \pnt{r}$. Intuitively, the clause $C$ should
remain redundant in the \ti{smaller} subspace specified by
\pnt{r}. However, this is not the case.  $\prob{X}{\cof{F}{q}} \equiv
\prob{X}{\cof{F}{q} \setminus \s{\cof{C}{q}}}$ does not imply
$\prob{X}{\cof{F}{r}} \equiv \prob{X}{\cof{F}{r} \setminus
  \s{\cof{C}{r}}}$ (see the example below).
\begin{example}
Let $G(X)$ be a satisfiable formula.  Then every $G^* \subseteq G$ is
redundant in \prob{X}{G}. Assume that $G$ is unsatisfiable in subspace
\pnt{r}. Then there is $G^* \subseteq G$ that is \ti{not} redundant in
\prob{X}{G} in subspace \pnt{r}. So $G^*$ is redundant in subspace
$\pnt{q}=\emptyset$ and is not redundant in subspace \pnt{r} where
$\pnt{q} \subset \pnt{r}$. An obvious problem here is that the formula
\cof{G}{r} does not preserve the structure of $G$ (in terms of
redundancy of clauses).
\end{example}

The problem above can be easily addressed by using a more
sophisticated notion of redundancy called \ti{virtual
  redundancy}~\cite{qe_learn}. (The latter is different from
redundancy specified by Definition~\ref{def:red_cls}.) However, this
would require adding more definitions and propositions. So, for the
sake of simplicity, in this paper, we assume that if $C$ is redundant
in \prob{X}{F} in a subspace \pnt{q} (where redundancy is specified by
Definition~\ref{def:red_cls}), it is also redundant in subspace
\pnt{r} if $\pnt{q} \subset \pnt{r}$ holds.

\section{Generation Of New D-sequents By Substitution}
\label{app:gen_new_dseqs}
In Section~\ref{sec:join_upd}, we recalled two methods of producing
new D-sequents. In this appendix, we describe one more procedure for
generating a new D-sequent. This procedure is to reshape the structure
constraint $H$ of a D-sequent by substituting a clause $C \in H$ with
the structure constraint of another D-sequent stating redundancy of
$C$.
%
%
\begin{proposition}
  \label{prop:repl}
Let \prob{X}{F} be an \ecnf formula. Let \oDs{\pnt{q'}}{H'}{C'} and
$(\pnt{q''}$, $H'')$ $\rightarrow C''$ be consistent D-sequents (see
Definition~\ref{def:cons_dseqs}). Let $C''$ be in $H'$. Then the
D-sequent \ods{q}{H}{C'} holds where $\pnt{q} = \pnt{q'} \cup
\pnt{q''}$,  $H\!=\!(H'\! \setminus\! \s{C''}) \cup H''$.
\end{proposition}

The proposition below is an implication of Proposition~\ref{prop:repl}
to be used in Appendix~\ref{app:str_con}. Since this proposition is
not mentioned in~\cite{qe_learn}, we prove it here.
%
%
\begin{proposition}
\label{prop:dseq_subsp}
Let \prob{X}{F} be an \ecnf formula. Let $S$ be a D-sequent equal to
\ods{q}{H}{C}. Let \pnt{r} be an assignment satisfying at least one
clause of $H$ where $\pnt{q} \subset \pnt{r}$ holds. Then the
D-sequent $S'$ equal to \ods{r}{H'}{C} holds where $H'$ is obtained
from $H$ by removing the clauses satisfied by \pnt{r}.
\end{proposition}
\begin{proof}
  Let $C_1,\dots,C_k$ be the clauses of $H$ satisfied by \pnt{r}.  For
  each clause $C_i$, $1 \leq i \leq k$, one can build an atomic
  D-sequent $S_i$ of the first kind equal to $(\ppnt{q}{i},H_i)
  \rightarrow C_i$. Here $H_i = \emptyset$ and $\ppnt{q}{i} = (v_i =
  b_i)$ is an assignment satisfying $C_i$ where $\ppnt{q}{i} \subseteq
  \pnt{r}$. It is not hard to show that $S,S_1$ are consistent (see
  Definition~\ref{def:cons_dseqs}).  By applying
  Proposition~\ref{prop:repl}, one obtains a new D-sequent $S'$ equal
  to \ods{q'}{H'}{C} where $\pnt{q'} = \pnt{q} \cup \ppnt{q}{1}$ and
  $H' = H \setminus \s{C_1}$. Since the new D-sequent $S'$ is
  consistent with $S_2$, one can apply Proposition~\ref{prop:repl}
  again. Going on in such a manner, one obtains the D-sequent $S'$
  equal to \ods{r'}{H'}{C} where $\pnt{r'} = \pnt{q} \cup \ppnt{q}{1}
  \dots \cup \ppnt{q}{k}$ and $H' = H \setminus
  \s{C_1,\dots,C_k}$. Since $\pnt{r'} \subseteq \pnt{r}$, the
  D-sequent \ods{r}{H'}{C} holds as well (see
  Appendix~\ref{app:red_subsp})
\end{proof}

\section{Examples Of D-sequent Generation By \lrn}
\label{app:dseq_gen}
In this appendix, we give examples of how \lrn builds D-sequents. (The
only exception is Appendix~\ref{ssec:cnfl_cls} where we describe how a
conflict clause generated.) We continue using the notation of
Section~\ref{sec:Apqe}. In particular, we assume that \Apqe is applied
to take $F_1$ out of the scope of quantifiers in \prob{X}{F_1(X,Y)
  \wedge F_2(X,Y)}.

%
%
\subsection{Generation of a conflict clause}
\label{ssec:cnfl_cls}
Suppose that a clause \Sub{C}{fls} of $F_1 \wedge F_2$ is falsified in
\ti{BCP} (see Fig.~\ref{fig:bcp}). Let \pnt{a} be the current
assignment at the point of \ti{BCP} where \Sub{C}{fls} gets
falsified. (So \pnt{a} falsifies \Sub{C}{fls}). Then \lrn generates a
\ti{conflict clause} i.e.  a clause falsified by \pnt{a} in which
there is only one literal falsified at the conflict level.  This
clause is built as follows~\cite{grasp}. First, \lrn picks the last
literal \ti{Lit} of \Sub{C}{fls} falsified by \pnt{a}. \lrn takes the
clause $C$ from which the value falsifying \ti{Lit} was derived and
resolves it with \Sub{C}{fls} producing a new clause \Sub{C}{fls}
falsified by \pnt{a}.  Then \lrn again picks the last literal \ti{Lit}
of \Sub{C}{fls} falsified by \pnt{a}. This goes on until only one
literal of \Sub{C}{fls} is falsified by an assignment made at the
conflict level and this is the \ti{decision} assignment of this
level. At this point, \Sub{C}{fls} is a conflict clause that is added
to $F_1 \wedge F_2$.

 %
 %
\subsection{D-sequent generation when  current target is satisfied}
\label{ssec:sat_trg}
 Suppose that $Y = \s{y}$ and $F_1 \wedge F_2$ contains (among others)
 the clauses $C_1 = y \vee x_1$, $C_2 = x_1 \vee x_2$.  Suppose $C_2$
 is the current target clause and \prd makes the decision assignment
 $y=0$. \ti{BCP} finds out that $C_1$ is unit and derives the
 assignment $x_1=1$ satisfying $C_2$. So \ti{BCP} terminates reporting
 the backtracking condition \ti{SatTrg} (line 11 of
 Fig.~\ref{fig:bcp}). At this point the current assignment \pnt{a} is
 equal to $(y=0,x_1=1)$. Then \lrn builds a D-sequent $S$ as follows.
 It starts with the atomic D-sequent $S$ of the first kind equal to
 \ods{q}{H}{C_2} where $\pnt{q} = (x_1=1)$ and $H = \emptyset$. The
 D-sequent $S$ states that $C_2$ is satisfied by \pnt{q} and hence
 redundant in the subspace \pnt{q} (and so in the subspace
 \pnt{a}). The conditional \pnt{q} contains the assignment $(x_1=1)$
 derived from clause $C_1$.  \lrn gets rid of this assignment as
 described in Subsection~\ref{ssec:dseq_gen}. First, it forms the
 D-sequent $S'$ equal to \ods{q'}{H'}{C_2} where $\pnt{q'} =
 (y=0,x_1=0)$ and $H' = \s{C_1}$. This is an atomic D-sequent of the
 second kind stating the redundancy of $C_2$ in the subspace where
 $C_1$ is falsified. Then \lrn joins $S$ and $S'$ at variable $x_1$ to
 obtain a new D-sequent $S$ equal to \ods{q}{H}{C_2} where \pnt{q} =
 $(y=0)$ and $H = \s{C_1}$.  The conditional \pnt{q} of $S$ does not
 contain assignments \ti{derived} at the current decision level. So
 \lrn terminates returning $S$.

%
%
\subsection{D-sequent generation when  current target is blocked}
Suppose that $Y = \s{y}$ and $F_1 \wedge F_2$ contains (among others)
the clauses $C_1 = y \vee x_1$, $C_2 = x_1 \vee \overline{x}_2$, $C_3
= x_2 \vee x_3$. Suppose that $C_3$ is the current target clause and\
$C_2$ is the only clause of $F_1 \wedge F_2$ that can be resolved with
$C_3$ on $x_2$.

Suppose \prd made the assignment $y = 0$. By running \ti{BCP}, \prd
derives the assignment $x_1 = 1$ from clause $C_1$. This assignment
satisfies $C_2$, which makes the target clause $C_3$ blocked at
$x_2$. At this point, \lrn generates a D-sequent as follows. First, an
atomic D-sequent $S$ of the third kind is generated (see
Subsection~\ref{ssec:third_kind}). $S$ is equal to \ods{q}{H}{C_3}
where $\pnt{q} = (x_1=1)$, $H = \emptyset$.

The conditional \pnt{q} of $S$ contains the assignment $x_1\!=\!1$
derived at the current decision level. To get rid of it, \lrn joins
$S$ with the D-sequent $S' = $\ods{q'}{H'}{C_3} at $x_1$ where
$\pnt{q'}=(y=0,x_1=0)$, $H' = \s{C_1}$. This D-sequent states
redundancy of $C_3$ in the subspace where $C_1$ is falsified. After
joining $S$ and $S'$, one obtains a new D-sequent $S$ equal to
\ods{q}{H}{C_3} where $\pnt{q} = (y=0)$, $H = \s{C_1}$. The
conditional of $S$ does not contain assignments derived at the current
decision level. So $S$ is the final D-sequent returned by \lrn.

%
%
\subsection{D-sequent generation when a non-target clause is falsified}
Generation of a conflict clause $C$ implies that every literal
\ti{Lit} resolved out in the process of obtaining $C$ is falsified by
the assignment derived from a \ti{clause} (see
Appendix~\ref{ssec:cnfl_cls}).  Suppose that at least one such an
assignment is derived from a \ti{D-sequent}. Then \lrn \ti{cannot}
derive a clause implied by $F_1 \wedge F_2$ and falsified by the
current assignment \pnt{a}.  Instead, \lrn derives a \ti{D-sequent}
stating redundancy of the current target clause \Sub{C}{trg}. In this
subsection, we consider the case where the clause \Sub{C}{fls}
falsified by \ti{BCP} (i.e. the starting point of \lrn) is different
from \Sub{C}{trg}.  The next subsection considers the case where
\Sub{C}{fls}=\Sub{C}{trg}.

Suppose that $Y = \s{y}$ and $F_1 \wedge F_2$ contains (among others)
the clauses $C_1 = y \vee \overline{x}_1 \vee x_2$, $C_2 =
\overline{x}_1 \vee x_3$, $C_3= \overline{x}_2 \vee
\overline{x}_3$. Suppose that the current target clause is $C_4 \in
(F_1 \cup F_2)$. Suppose that the D-sequent $S^*$ equal to
\ods{q^*}{H^*}{C_4} was derived earlier where $\pnt{q^*} =
(y=0,x_1=0)$ and $H^* = \emptyset$. Assume that \prd makes the
decision assignment $y=0$. By running \ti{BCP}, \prd first derives
$x_1 = 1$ from $S^*$. This is due to the fact that $S^*$ becomes unit
under $y=0$ and $x_1=1$ \ti{deactivates} $S^*$ (see
Subsection~\ref{ssec:sngl_dseq_reuse}).  Then \prd derives $x_2=1$ and
$x_3=1$ from $C_1$ and $C_2$ respectively.  These two assignments
falsify $C_3$.

The final D-sequent stating redundancy of $C_4$ is built by \lrn as
follows. First, \lrn generates the D-sequent $S$ equal to
\ods{q}{H}{C_4} where $\pnt{q} = (x_2=1,x_3=1)$, $H = \s{C_3}$. It is
an atomic D-sequent of the second kind stating redundancy of $C_4$ in
the subspace where $C_3$ is falsified. The conditional \pnt{q} of $S$
contains assignments derived at the current decision level. So, \lrn
picks the most recent derived assignment of \pnt{q} i.e. $x_3=1$ and
gets rid of it. This is achieved by joining $S$ with the D-sequent
$S'$ equal to \ods{q'}{H'}{C_4} at variable $x_3$ where $\pnt{q'} =
(x_1=1,x_3=0)$, $H' = \s{C_2}$. The D-sequent $S'$ states redundancy
of $C_4$ in the subspace where $C_2$ is falsified.  The result of
joining $S$ and $S'$ is a new D-sequent $S$ equal to \ods{q}{H}{C_4}
where $\pnt{q}=(x_1=1,x_2=1)$, $H = \s{C_2,C_3}$.

\lrn again picks the most recent derived assignment of \pnt{q}
i.e. $x_2=1$.  Then it joins $S$ with the D-sequent $S''$ equal to
\ods{q''}{H''}{C_4} at variable $x_2$ where $\pnt{q''} =
(y=0,x_1=1,x_2=0)$, $H'' = \s{C_1}$. The D-sequent $S''$ states
redundancy of $C_4$ in the subspace where $C_1$ is falsified.  The
result of joining $S$ and $S''$ is a new D-sequent $S$ equal to
\ods{q}{H}{C_4} where $\pnt{q}=(y=0,x_1=1)$, $H = \s{C_1,C_2,C_3}$.

Finally, \lrn gets rid of the assignment $x_1=1$ derived from the
D-sequent $S^*$ above. To this end, \lrn joins $S$ with $S^*$ at
variable $x_1$ to produce the D-sequent $S$ equal to \ods{q}{H}{C_4}
where $\pnt{q}=(y=0)$, \mbox{$H = \s{C_1,C_2,C_3}$}. This is the final
D-sequent $S$ returned by \lrn.

%
%
\subsection{D-sequent generation when a target clause is falsified}
\label{ssec:dseq_and_clause}
In this subsection, we continue the topic of the previous subsection.
Here, we consider the case \Sub{C}{fls}=\Sub{C}{trg} i.e.  the current
target clause is falsified by \ti{BCP}. (As in the previous
subsection, we assume that at least one assignment that ``matters'' is
derived from a D-sequent rather than a clause.) Then \lrn still
derives a D-sequent $S$ stating redundancy of \Sub{C}{trg} but also
adds a new clause $C$. The latter is not a full-fledged conflict
clause: it may contain \ti{more than one literal} falsified at the
conflict level. \lrn has to add $C$ to $F_1 \wedge F_2$ to make
\Sub{C}{trg} redundant. \lrn derives $S$ and $C$ as follows. The
clause $C$ is built similarly to a conflict clause until \lrn reaches
an assignment derived from a D-sequent. Then \lrn builds $S$ by the
procedure described in the previous subsection where $C$ is used as a
``starting clause'' falsified by \pnt{a}.

Let us re-examine the example of the previous subsection under the
assumption that the falsified clause $C_3$ is also the current target
clause. \lrn resolves $C_3$ with $C_2$ (on variable $x_3$) and $C_1$
(on variable $x_2$) to produce the clause $C = y \vee
\overline{x}_1$. This clause is falsified by the current assignment
\pnt{a}.  Then \lrn builds the D-sequent $S$ equal to \ods{q}{H}{C_3}
where $\pnt{q} = (y=0,x_1=1)$ and $H = \s{C}$. This D-sequent states
the redundancy of $C_3$ in the subspace where the new clause $C$ is
falsified. Finally, \lrn gets rid of $x_1=1$ (derived from the
D-sequent $S^*$) in the conditional \pnt{q} of $S$. To this end, \lrn
joins $S$ with $S^*$ at variable $x_1$ to produce a new D-sequent $S$
equal to \ods{q}{H}{C_3} where $\pnt{q} = (y=0)$ and $H = \s{C}$. Then
\lrn terminates returning the D-sequent $S$ and clause $C$.

%
%
\subsection{D-sequent generation when a  D-sequent is activated}
In this subsection, we discuss the case where a D-sequent $S$ derived
earlier becomes active. This means that \Sub{C}{trg} is redundant in
the current subspace \pnt{a}. If the conditional of $S$ contains
assignments derived at the current decision level, \lrn generates a
new D-sequent whose conditional does not contain such assignments.
Consider the following example.  Let $Y = \s{y}$ and $F_1 \wedge F_2$
contain (among others) the clauses $C_1 = y \vee x_1$ and $C_2 = y
\vee x_2$.  Suppose that clause $C_3$ of $F_1 \wedge F_2$ is the
current target clause.  Suppose that the D-sequent $S$ equal to
\ods{q}{H}{C_3} was derived earlier where $\pnt{q} = (x_1=1,x_2=1)$
and $H = \emptyset$.

Assume that \prd made the decision assignment \mbox{$y=0$}. After
running \ti{BCP}, the assignments $x_1=1$ and $x_2=1$ are derived from
$C_1$ and $C_2$ respectively, which activates the D-sequent $S$. Note
that the conditional \pnt{q} of $S$ contains assignments derived at
the current level. So \lrn generates a new D-sequent as follows.
First $S$ is joined with the D-sequent $S'$ equal to \ods{q'}{H'}{C_3}
at variable $x_2$ where $\pnt{q'}=(y=0,x_2=0)$, $H' = \s{C_2}$. This
D-sequent states the redundancy of $C_3$ in the subspace where $C_2$
is falsified.  The resulting D-sequent $S$ is equal to \ods{q}{H}{C_3}
where $\pnt{q}=(y=0,x_1=1)$ and $H=\s{C_2}$.

Then $S$ is joined with the D-sequent $S''$ equal to
\ods{q''}{H''}{C_3} at variable $x_1$ where $\pnt{q''}=(y=0,x_1=0)$,
$H'' = \s{C_1}$. This D-sequent states the redundancy of $C_3$ in the
subspace where $C_1$ is falsified.  The resulting D-sequent $S$ is
equal to \ods{q}{H}{C_3} where $\pnt{q}=(y=0)$ and
$H=\s{C_1,C_2}$. The conditional of $S$ does not contain assignments
derived at the current decision level.  So $S$ is the final D-sequent
returned by \lrn.

\section{Updating Stack Of Target Levels}
\label{app:trg_lvls}
In this appendix, we give an example of how the stack $T$ of target
levels is updated.  Let \prob{X}{F(X,Y)} be an \ecnf formula where $F
= C_1 \wedge \dots \wedge C_5$. Here \mbox{$C_1 = y \vee x_1$}, $C_2 =
\overline{x}_1 \vee x_2$, $C_3=\overline{x}_2 \vee x_3 \vee x_4$, $C_4
= \overline{y} \vee \overline{x}_3$, $C_5= x_3 \vee \overline{x}_4$
and $Y = \s{y}$.  Consider the PQE problem of taking $C_1$ out of the
scope of quantifiers. So $C_1$ is the primary target clause and,
originally, $T$ has only one target level consisting of $C_1$.

Suppose that \prd makes the assignment $y=0$. The \ti{BCP} procedure
finds out that $C_1$ became a unit clause and adds $x_1=1$ derived
from $C_1$ to the assignment queue $Q$. Since $Q$ does not have any
other assignments to make and $C_1$ is the current target clause
\Sub{C}{trg}, the $\mi{BCP}^*$ procedure is called (line 4 of
Fig.~\ref{fig:bcp}). It makes the assignment $x_1=1$ derived from
$C_1$ and creates a new top level of $T$. This level is specified by
the pair $(C_1,x_1)$ where $C_1$ is the key clause and $x_1$ is the
key variable of this level. The latter consists of $C_2$, the only
clause of $F$ resolvable with $C_1$ on $x_1$. Then $\mi{BCP}^*$
derives $x_2=1$ from $C_2$ and creates a new top level of $T$
specified by the pair $(C_2,x_2)$. This level consists of $C_3$, the
only clause of $F$ resolvable with $C_2$ on $x_2$.

At this point, $\mi{BCP}^*$ runs out of unit clauses and returns to
\ti{BCP}. Then \ti{BCP} uses the top level of $T$ to pick the next
target clause $T$.  Since the top level of $T$ consists only of $C_3$,
the latter is chosen as \Sub{C}{trg} (line 5) and \ti{BCP} breaks the
while loop (line 8).

\section{Special Backtracking}
\label{app:spec_bcktr}
In this appendix, we discuss backtracking performed by \prd when $T$
contains secondary targets. Let us continue considering the example of
Appendix~\ref{app:trg_lvls}. After picking $C_3$ as the current target
clause, \ti{BCP} finds out that $C_3$ is blocked (line 18 of
Fig.~\ref{fig:bcp}) at variable $x_3$. Indeed, $C_4$ is the only
clause with $\overline{x}_3$ and it is satisfied by $y=0$. So \ti{BCP}
returns the backtracking condition \ti{BlkTrg} (line 19). Then \prd
learns an atomic D-sequent $S$ of the third kind (line 11 of
Fig.~\ref{fig:prv_red}) equal to \ods{q}{H}{C_3} where $\pnt{q} =
(y=0)$, $H = \emptyset$.

Since $C_3$ is a secondary target, \prd skips the regular backtracking
part (lines 14-17) and goes to the third part of the while loop (lines
18-25).  Recall that the current assignment \pnt{a} is
$(y=0,x_1=1,x_2=1)$ and $T$ consists of three target levels. The
bottom level of $T$ consists of the primary target $C_1$. Then next
level is specified by the pair $(C_1,x_1)$ and the top level of $T$ is
specified by the pair $(C_2,x_2)$.

At this point, \prd calls the special backtracking procedure
\ti{SpecBcktr} (line 18).  Although the conditional of $S$ contains
only assignment to $y$, \ti{SpecBcktr} cannot undo assignments $x_1=1$
and $x_2=1$ for the reason explained in
Subsection~\ref{ssec:spec_bcktr}. The clause $C_3$ (whose redundancy
is stated by $S$) became a secondary target \ti{only} after the
assignment $x_2=1$ was made. Since there is no conflict,
\ti{SpecBcktr} cannot backtrack past $x_2=1$ (i.e. the ``point of
origin''). So, \ti{SpecBcktr} terminates without changing \pnt{a}.

Since the conditional of $S$ does not contain any assignments made
after $x_2=1$, the redundancy of $C_3$ is proved up to the point of
origin. So \prd calls \ti{NewTrg} to pick a new target among the
clauses of the top level of $T$ (line 22). As we mentioned earlier,
the current top level of $T$ consists only of $C_3$. So no new target
can be found. As we mentioned in Subsection~\ref{ssec:spec_bcktr},
this means that $C_2$, the key clause of the top level of $T$, is
blocked at variable $x_2$. (Because the only clause of $F$ resolvable
with $C_2$ on $x_2$ is proved redundant.)  So, \ti{NewTrg} generates
an atomic D-sequent $S'$ of the third kind stating the redundancy of
$C_2$.  This D-sequent is equal to \ods{q'}{H'}{C_2} where
$\pnt{q'}=(y=0)$ and $H' = \emptyset$. Then \ti{NewTrg} eliminates the
current top level of $T$ restoring the clause $C_3$ back into the
formula $F$ (i.e. treating it as present in formula $F$). Besides,
\ti{NewTrg} unassigns $x_2$.

Now, \ti{NewTrg} tries to find a target clause in the \ti{new} top
level of $T$ specified by the pair $(C_1,x_1)$. Since the only clause
of $F$ resolvable with $C_1$ on $x_1$ (i.e. $C_2$) is proved redundant
in the subspace \pnt{a}, \ti{NewTrg} repeats the actions described
before.  First, it derives an atomic D-sequent $S''$ of the third kind
stating the redundancy of $C_1$. Here $S''$ is equal to
\ods{q''}{H''}{C_2} where $\pnt{q''}=(y=0)$ and $H'' =
\emptyset$. Then it eliminates the top level of $T$ restoring $C_2$
back into formula $F$ and unassigning $x_1$.

Now, $T$ is reduced to the primary target level (containing clause
$C_1$).  \ti{NewTrg} terminates returning $C_1$ as \Sub{C}{trg} and
$S''$ as D-sequent $S$ (line 22). Since $C_1$ is the primary target,
\prd checks if the conditional of $S$ is empty (line 24). Since it is
not, the redundancy of $C_1$ is proved only in the subspace \pnt{a}
that is currently equal to $(y=0)$. Since $S$ is unit in the subspace
\pnt{a}, \prd adds the assignment $y=1$ derived from $S$ to the
assignment queue $Q$ (line 25). Then \prd starts a new iteration of
the while loop.

\section{Correctness of \Apqe}
\label{app:sound_compl}
In this appendix, we give an informal proof that \Apqe is sound and
complete.
%
%
\subsection{\Apqe is sound}
\label{ssec:sound}
Let \prob{X}{F_1(X,Y) \wedge F_2(X,Y)} be an \ecnf formula. Suppose
that \Apqe is used to take $F_1$ out of the scope of
quantifiers. Assume, for the sake of simplicity, that every clause of
$F_1$ is an $X$-clause. In its operation, \Apqe generates new clauses
obtained by resolving clauses of $F_1 \wedge F_2$ and thus implied by
$F_1 \wedge F_2$. The final formula produced by \Apqe can be
represented as \prob{X}{F} where $F = \Sup{F_1}{ini} \wedge F^*_1
\wedge F^{**}_1 \wedge \Sup{F_2}{ini} \wedge F^*_2$.  Here
\Sup{F_1}{ini} and \Sup{F_2}{ini} denote the initial versions of $F_1$
and $F_2$ respectively.  The formula $F^*_1(Y)$ denotes the derived
clauses depending only on variables of $Y$ whose generation involved
clauses of \Sup{F_1}{ini} and/or their descendants. The formula
$F^{**}_1(X,Y)$ denotes the derived $X$-clauses whose generation
involved clauses of \Sup{F_1}{ini} and/or their descendants. The
formula $F^*_2(X,Y)$ denotes the derived clauses whose generation
involved \ti{only} clauses of \Sup{F_2}{ini} and/or their descendants.

For every clause $C$ of $\Sup{F_1}{ini} \wedge F^{**}_1$, \Apqe calls
\prd that generates a D-sequent $S$ equal to \ods{q}{H}{C} where
$\pnt{q} = \emptyset$ and $H \subseteq (F \setminus \s{C})$.  The
D-sequent $S$ states redundancy of $C$ in formula \prob{X}{F}. This
D-sequent is correct due to correctness of the atomic D-sequents and
the join operation and due to Proposition~\ref{prop:dseq_add_cls} of
Appendix~\ref{app:gen_new_dseqs}. The D-sequents for the clauses of
$\Sup{F_1}{ini} \wedge F^{**}_1$ are derived by \Apqe one by one in
some order. That is these D-sequents are \ti{consistent} (see
Definition~\ref{def:cons_dseqs}). So, one can claim that
\prob{X}{\Sup{F_1}{ini} \wedge \Sup{F_2}{ini}}~$\equiv$ \prob{X}{F}
$\equiv$ $F^*_1 \wedge$ \prob{X}{\Sup{F_2}{ini} \wedge F^*_2} $\equiv
F^*_1 \wedge$ \prob{X}{\Sup{F_2}{ini}}. Thus, $F^*_1(Y)$ is indeed a
solution to the PQE problem at hand.

%
%
\subsection{\Apqe is complete}
Assume that \Apqe generates only clauses that have not been seen
before.  Taking into account that \Apqe examines a finite search tree,
this means that \Apqe always terminates and thus is complete. The
problem however is that the version of \prd described in
Section~\ref{sec:Apqe} \ti{may} generate a duplicate of an $X$-clause
that is currently \ti{proved redundant}. So \prd and hence \Apqe may
loop.

To prevent looping, the current implementation of \prd does the
following.  (For the sake of simplicity, we did not discuss this part
of \prd in Section~\ref{sec:Apqe}.) Let (\pnt{y},\pnt{x}) be the
current assignment to $Y \cup X$ made by \prd before a duplicate of an
$X$-clause is generated. After a duplicate $C$ is generated, \prd
discards $C$ and backtracks to the last assignment to a variable of
$Y$ (and thus undoing all assignments to $X$). This is accompanied by
removing all secondary targets from the stack $T$. So, on completion
of backtracking, the primary target clause \Sub{C}{pr} is the current
target clause.  Then \prd generates a D-sequent stating the redundancy
of \Sub{C}{pr} in subspace \pnt{y} and keeps going as if \prd just
finished line 11 of Figure~\ref{fig:prv_red}.

To generate the D-sequent above, \prd does the following.  First, \prd
runs an internal SAT-solver to check if formula $F$ (defined in the
previous subsection) is satisfiable in subspace \pnt{y}. If not, a
clause $C(Y)$ implied by $F$ is generated and added to $F$. Then an
atomic D-sequent of the second kind is generated stating the
redundancy of \Sub{C}{pr} in the subspace where $C(Y)$ is
falsified. Otherwise, \prd finds an assignment (\pnt{y},\pnt{x})
satisfying $F$. The existence of such an assignment means that
\Sub{C}{pr} is redundant in the subspace \pnt{y} without adding any
clauses. Then \prd generates a D-sequent \ods{y^*}{H}{\Sub{C}{pr}}
where $H = \emptyset$ and \pnt{y^*} $\subseteq \pnt{y}$ and
(\pnt{y^*},\pnt{x}) satisfies $F$. (In other words, \prd tries to
shorten the satisfying assignment to reduce the conditional of the
D-sequent constructed for \Sub{C}{pr})

\section{Reducing Size Of Structure Constraints}
\label{app:str_con}
In this appendix, we describe some methods for reducing the size of
structure constraints in D-sequents learned by \Apqe. Let $S$ be a
D-sequent \ods{q}{H}{C} stating redundancy of $C$ in \prob{X}{F_1(X,Y)
  \wedge F_2(X,Y)}.  A useful observation here is that one does not
need to keep a clause $B \in H$ if $\V{B} \subseteq Y$.  Indeed, \Apqe
proves redundancy only of $X$-clauses. So a clause $B(Y)$ is either
present in $H$ or is satisfied by the current assignment \pnt{a}. In
either case, $S$ can be safely reused in the subspace \pnt{a} (see
Proposition~\ref{prop:dseq_subsp}).

As we mentioned earlier, one can keep structure constraints small by
generating only D-sequents for the target clauses of the $k$ bottom
levels of the stack $T$ (see Section~\ref{sec:exper}). In this case,
the size of $H$ is limited by the total number of secondary target
clauses of levels $1,\dots,k$.  In particular, one can safely reuse
the D-sequents of the primary target clause (level 0 of $T$) without
computing structure constraints at all\footnote{\label{fn:str_con}
Let \ods{q}{H}{\Sub{C}{trg}} be a D-sequent $S$ where \Sub{C}{trg} is
the current target. Suppose $\pnt{q} \subseteq \pnt{a}$ holds
where \pnt{a} is the current subspace examined by \Apqe.
If \Sub{C}{trg} is the primary target, no $X$-clause of $H$ is a
secondary target (because the set of secondary targets is currently
empty). So a clause of $H$ is either present in $F_1 \wedge F_2$
or \ti{satisfied} by \pnt{a}. In either case, $S$ can be safely reused
in the subspace \pnt{a} (see Proposition~\ref{prop:dseq_subsp}).
}.

Another method is based on using the substitution operation (see
Appendix~\ref{app:gen_new_dseqs}). By repeatedly applying this
operation, one can reduce the structure constraint of a D-sequent $S$
to an empty set (which may increase the conditional of
$S$)~\cite{qe_learn}.

Finally, one can reduce the size of structure constraints by adding
new clauses. Consider the following example.  Suppose that
\prob{X}{F_1(X,Y) \wedge F_2(X,Y)} contains (among others) clauses
$C_1 = y \vee x_1$, $C_2 = \overline{x}_1 \vee x_2$, ..., $C_k =
\overline{x}_{k-1} \vee x_k$, $C_{k+1}=x_k \vee x_{k+1}$.  Suppose
$C_{k+1}$ is the current target clause. Suppose that $Y = \s{y}$ and
\prd made the assignment $y=0$. After running BCP, \prd derives
$x_2=1$ from $C_2$, $x_3=1$ from $C_3$ and so on until \mbox{$x_k=1$}
is derived from $C_k$. The latter assignment satisfies the target
clause $C_{k+1}$. In this case, in our current implementation, the
\lrn procedure generates the D-sequent $S$ equal to
\ods{q}{H}{C_{k+1}} where $\pnt{q}=(y=0)$ and $H =
\s{C_1,\dots,C_k}$. (This D-sequent is constructed as described in
Appendix~\ref{ssec:sat_trg} where one performs $k$ join operations at
variables $x_k,\dots,x_1$.) Note that no \ti{new} clauses are added
when building $S$. One can reduce the size of $H$ by generating a new
clause $C=y \vee x_k$ obtained by resolving clauses $C_k,\dots,C_1$ on
variables $x_k,\dots,x_1$. Adding $C$ to $F_1 \wedge F_2$ makes
$C_{k+1}$ redundant in subspace $y=0$. So one can derive the D-sequent
\ods{q}{H}{C_{k+1}} where $\pnt{q}=(y=0)$, $H = \s{C}$ stating redundancy of
$C_{k+1}$ in \prob{X}{C \wedge F_1 \wedge F_2}. Thus, one reduces the
size of the structure constraint $H$ at the expense of adding a new
clause.

\section{One More Experiment}
\label{app:exper}
In Section~\ref{sec:exper}, we describe some experiments with \Apqe.
In this appendix, we describe one more experiment.  Similarly to the
experiments of Section~\ref{sec:exper}, \Apqe stored/reused only
D-sequents of the primary target clauses. This could not affect the
results of Section~\ref{sec:exper} much because, for the problem we
considered there, generation of a secondary target clause was a rare
event.  This was not true for the problem considered in this
appendix where a very large number of secondary targets was
generated. Nevertheless, the experimental results presented in this
appendix show that reusing even a small fraction of D-sequents can be
beneficial.

In this section, we solve the PQE problem arising in the method of
equivalence checking introduced in~\cite{fmcad16}. (The main idea of
this method is sketched in Subsection~\ref{ssec:pqe_ec}.) Let
$N'(X',Y',z')$ and $N''(X'',Y'',z'')$ be single-output circuits to be
checked for equivalence. Let \Sub{T}{cut} specify the variables of a
cut in $N'$ and $N''$.  Let \Sub{F}{cut} specify the gates of $N'$ and
$N''$ located between the inputs and the cut. Let \Sub{W}{cut} denote
$\V{\Sub{F}{cut}} \setminus \Sub{T}{cut}$. Let $\mi{EQ}(X',X'')$ be a
formula evaluating to 1 iff $X'$ an $X''$ have the same assignments.
We consider the problem of taking $\mi{EQ}$ out of the scope of
quantifiers in \prob{\Sub{W}{cut}}{\mi{EQ} \wedge \Sub{F}{cut}}. That
is one needs to find a formula $G(\Sub{T}{cut})$ such that $G \wedge
\prob{\Sub{W}{cut}}{\Sub{F}{cut}}$ $\equiv$
\prob{\Sub{W}{cut}}{\mi{EQ} \wedge \Sub{F}{cut}}.

%
%
\begin{table}
  \small
  \caption{Computing cut constrains in equivalence checking. The time limit is
  set to 100 seconds}
\scriptsize
\begin{center}
\begin{tabular}{|p{40pt}|p{12pt}|p{22pt}|p{10pt}|p{14pt}|p{10pt}|p{14pt}|p{10pt}|} \hline
    name     &$k$& \multicolumn{2}{c|}{\apqe} & \multicolumn{2}{c|}{\Apqe}& \multicolumn{2}{c|}{\Apqe}  \\
          & &\multicolumn{2}{c|}{} & \multicolumn{2}{c|}{no learning} & \multicolumn{2}{c|}{limited learning}\\ \cline{3-8}
         & & \#dseqs      & time      &\#dseqs& time   &\#dseqs & time    \\
         &     & $\times 10^3$ &(s.)  & $\times 10^3$    & (s.)  & $\times 10^3$ & (s.) \\ \hline

kenoopp1       &  5    &  $>$22,500     &  $*$    &  4 &  0.3      &  2.9 &   \tb{0.2}     \\\hline
bj08autg3f1    &  50   &  424    &  0.8    &  15  &  0.6      &  7.2 &   \tb{0.3}     \\\hline
abp4pold       &  2    &  $>$18,786     &  $*$    &  25  &  0.9      & 7 &  \tb{0.2}    \\\hline
eijks838       &  5    &  302    &  \tb{0.6}    &  43  &  1.0      &  52  &   1.1     \\\hline
pdtpmstwo      &  3    &  1,345    &  2.1    &  310   &  5.1      &  96   &   \tb{1.9}    \\\hline
prodconspold4  &  15   &  89   &  0.2    &  1.7  &  \tb{0.02}    &  2  &   \tb{0.02}   \\\hline
kenoopp2       &  2    &  $>$12,985    &  $*$    &  4.4 &  1.3      &  3.3 &   \tb{1.1}    \\\hline
texastwoprocp1 &  12   &  465    &  0.8    &  68 &  1.0      &  17  &   \tb{0.3}    \\\hline
eijks953       &  3    &  331    &  2.0    &  1 &  0.06     &  0.6  &   \tb{0.05}   \\\hline
bj08amba2g1    &  3    &  2,414    &  6.0    &  3.5 &  0.6      &  1.7 &   \tb{0.3}    \\\hline
pdtvishuffman1 &  40   &  1,656    &  1.8    &  1 &  0.3      &  0.6 &   \tb{0.2}    \\\hline
pdtpmsheap     &  3    &  2,039    &  6.0    &  3.7 &  0.08     &  2.1 &   \tb{0.07}   \\\hline
pdtvistimeout1 &  16   &  22,902     &  39     &  8.9 &  0.03     &  8 &   \tb{0.02}   \\\hline
brpptimo       &  2    &  $>$16,841     &  $*$    &  9.1 &  1.0      &  4.1 &   \tb{0.4}    \\\hline

\end{tabular}                
\end{center}
\label{tbl:ec}
\vspace{-5pt}
\end{table}


In this experiment, we used HWMCC-10 benchmarks. Circuit $N'$ was
specified by the transition relation of a benchmark and $N''$ was
obtained from $N'$ by a logic optimization tool. A cut of $N'$ and
$N''$ was formed from gates located in $N'$ and $N''$ at topological
levels $\leq k$. (The set consisting only of gates of topological level
$k$, in general, does not form a cut.) The input variables have
topological level 0. Table~\ref{tbl:ec} shows results for a sample of
the set of benchmarks we tried.

The first column gives the name of a benchmark. The second column
gives the value of $k$ above. The following columns give the number of
generated D-sequents (in thousands) and the run time for the PQE
procedures we compared. In this experiment, we compared the same
procedures as in Tables~\ref{tbl:mc}. The results of
Table~\ref{tbl:ec} show that \Apqe with limited learning outperforms
\apqe and \Apqe without learning.

\end{document}